\documentclass[english]{article}
\usepackage{lmodern}

\usepackage[T1]{fontenc}
\usepackage[latin9]{inputenc}
\usepackage{geometry}
\usepackage{color}
\usepackage{babel}
\usepackage{verbatim}
\usepackage{float}
\usepackage{units}
\usepackage{amsmath}
\usepackage{amsthm}
\usepackage{amssymb}
\usepackage{graphicx}
\usepackage[authoryear]{natbib}
\usepackage[unicode=true,pdfusetitle,
 bookmarks=true,bookmarksnumbered=false,bookmarksopen=false,
 breaklinks=false,pdfborder={0 0 1},backref=page,colorlinks=true]
 {hyperref}
\usepackage{breakurl}

\makeatletter

\floatstyle{ruled}
\newfloat{algorithm}{tbp}{loa}
\providecommand{\algorithmname}{Algorithm}
\floatname{algorithm}{\protect\algorithmname}

 \theoremstyle{definition}
  \newtheorem{example}{\protect\examplename}
  \theoremstyle{plain}
  \newtheorem{lemma}{\protect\lemmaname}
\theoremstyle{plain}
\newtheorem{theorem}{\protect\theoremname}
\theoremstyle{plain}
\newtheorem{assumption}{\protect\assumptionname}
\theoremstyle{plain}
\theoremstyle{plain}
\newtheorem{property}{\protect\propertyname}
\theoremstyle{plain}
\theoremstyle{plain}
\newtheorem{proposition}{\protect\propositionname}
\theoremstyle{plain}

\usepackage{dsfont}
\hypersetup{pdftitle={Unbiased HMC},pdfauthor={Heng Jacob},linkcolor=RoyalBlue,citecolor=RoyalBlue}
\usepackage[dvipsnames,svgnames,x11names,hyperref]{xcolor}
\@ifundefined{showcaptionsetup}{}{%
 \PassOptionsToPackage{caption=false}{subfig}}
\usepackage{subfig}
\makeatother

\providecommand{\examplename}{Example}
\providecommand{\lemmaname}{Lemma}
\providecommand{\theoremname}{Theorem}
\providecommand{\assumptionname}{Assumption}
\providecommand{\propertyname}{Property}
\providecommand{\propositionname}{Proposition}

\begin{document}

\title{Unbiased Hamiltonian Monte Carlo with couplings}

\author{Jeremy Heng$^{*}$ and Pierre E. Jacob\thanks{Department of Statistics, Harvard University, USA. Emails: jjmheng@fas.harvard.edu \& pjacob@fas.harvard.edu.}}
\maketitle
\begin{abstract}
We propose a methodology to parallelize Hamiltonian Monte Carlo
estimators. Our approach constructs a pair of Hamiltonian Monte Carlo 
chains that are coupled in such a way that they meet exactly after some random number of iterations. 
These chains can then be combined so that resulting estimators are unbiased.
This allows us to produce independent replicates in parallel and average them 
to obtain estimators that are consistent in the limit of the number of replicates,
instead of the usual limit of the number of Markov chain iterations. 
We investigate the scalability of our coupling in high dimensions on a toy example. 
The choice of algorithmic parameters and the efficiency of our proposed methodology 
are then illustrated on a logistic regression with $300$ covariates, 
and a log-Gaussian Cox point processes model with low to fine grained discretizations. 
\end{abstract}
\textbf{\small{}Keywords}{\small{}: Coupling, Hamiltonian Monte Carlo, Parallel computing, Unbiased estimation.}{\small \par}

\section{Introduction \label{sec:Context-and-goal}}

\subsection{Parallel computation with Hamiltonian Monte Carlo \label{subsec:Parallelizing-Hamiltonian-Monte}}
Hamiltonian Monte Carlo is a Markov chain Monte Carlo method to approximate
integrals with respect to a target probability distribution $\pi$ on
$\mathbb{R}^{d}$.  Originally proposed by \citet{Duane:1987} in the physics
literature, it was later introduced in statistics by \citet{Neal:1993} and is
now widely adopted as a standard sampling tool
\citep{brooks2011handbook,lelievre2012}. Various aspects of its theoretical properties
have been studied: 
see \citet{betancourtetal:2017} and \citet{betancourt:2017} for its geometric properties, 
\citet{Livingstone:2016} and \citet{durmus2017convergence} for ergodicity results, 
\citet{beskos2013optimal}, \citet{Mangoubi:2017} and \citet{bourabee:2018}
for scaling results with respect to the dimension $d$. These results suggest 
that Hamiltonian Monte Carlo compares favorably to other Markov 
chain Monte Carlo algorithms such as random walk Metropolis--Hastings and 
Metropolis-adjusted Langevin algorithms in high dimensions.
In practice, Hamiltonian Monte Carlo is at the core of the No-U-Turn sampler
\citep{hoffman2014no} which is implemented in the software Stan
\citep{carpenter2016stan}. 

If one could initialize from the target distribution, usual estimators based on
any Markov chain Monte Carlo would be unbiased, and one could simply average
over independent chains \citep{rosenthal2000parallel}.  Except certain
applications where this can be achieved with perfect simulation methods
\citep{casella:lavine:robert:2001,huber2016perfect}, Markov
chain Monte Carlo estimators are ultimately consistent in the limit of the number of
iterations. Algorithms that rely on such asymptotics face the risk of becoming
obsolete if computational power continue to increase through the number of
available processors and not through clock speed. 

Several methods have been proposed to address this limitation with varying
generality \citep{mykland1995regeneration,neal2002circularly,glynn2014exact}.
Our approach builds upon recent work by \citet{jacob2017unbiased},
which introduces unbiased estimators based on Metropolis--Hastings algorithms and Gibbs samplers.
The present article describes how to design unbiased estimators 
for Hamiltonian Monte Carlo and some of its variants
\citep{girolami2011riemann}.  The proposed methodology is widely applicable and
involves a simple coupling between a pair of Hamiltonian Monte Carlo chains.
Coupled chains are run for a random but almost surely finite number of iterations, and
combined in such a way that resulting estimators are unbiased.  One can
produce independent copies of these estimators in parallel and average them
to obtain consistent approximations in the limit of the number of replicates.
This also yields confidence intervals valid in the number of replicates through
the central limit theorem; see also \citet{Glynn1991} for central limit theorems
parametrized by number of processors or time budget.
 
We begin by introducing some preliminary notation in Section \ref{subsec:Notation} and 
recapitulating the unbiased estimation framework of \citet{jacob2017unbiased} in Section \ref{subsec:Context:-unbiased-estimation}.  

\subsection{Notation\label{subsec:Notation}}

Given a sequence $(x_n)_{n\geq0}$ and integers $k< m$, we use the convention that $\sum_{n=m}^k x_n=0$.
The set of natural numbers is denoted by $\mathbb{N}$ and the set of non-negative real numbers by $\mathbb{R}_{+}$.
The $d$-dimensional vector of zeros is denoted by $0_{d}$ and the $d\times d$ identity matrix by $I_d$. 
The Euclidean norm of a vector $x\in\mathbb{R}^{d}$ is written as 
$|x|=(\sum_{i=1}^{d}x_{i}^{2})^{1/2}$. 
Given a subset $A\subseteq\varOmega$, the indicator function $\mathbb{I}_A:\varOmega\rightarrow\{0,1\}$ is defined as $\mathbb{I}_A(x)=1$ if $x\in A$, and $0$ if $x\in\varOmega\setminus A$. 
For a smooth function $f:\mathbb{R}^d\rightarrow\mathbb{R}$, we denote its gradient by 
$\nabla f:\mathbb{R}^d\rightarrow\mathbb{R}^d$ and its Hessian by $\nabla^2f:\mathbb{R}^d\rightarrow\mathbb{R}^{d\times d}$.
The gradient of a function $(x,y)\mapsto f(x,y)$ with respect to the variables $x$ and $y$
are denoted by $\nabla_{x}f$ and $\nabla_{y}f$ respectively.
Given functions $f:\mathbb{R}^n\rightarrow\mathbb{R}^m$ and $g:\mathbb{R}^d\rightarrow\mathbb{R}^n$, we define the 
composition $f\circ g:\mathbb{R}^d\rightarrow\mathbb{R}^m$ as $(f\circ g)(x)=f\{g(x)\}$ for all $x\in\mathbb{R}^d$.
The Borel $\sigma$-algebra of $\mathbb{R}^{d}$ is denoted by $\mathcal{B}(\mathbb{R}^{d})$; 
on the product space $\mathbb{R}^d\times\mathbb{R}^d$, $\mathcal{B}(\mathbb{R}^{d})\times\mathcal{B}(\mathbb{R}^{d})$ denotes the product $\sigma$-algebra.
The Gaussian distribution on $\mathbb{R}^d$ with mean vector $\mu$ and covariance matrix $\Sigma$
is denoted by $\mathcal{N}(\mu,\Sigma)$, and its density by $x\mapsto\mathcal{N}(x;\mu,\Sigma)$.
The uniform distribution on $[0,1]$ is denoted as $\mathcal{U}[0,1]$. We use the shorthand $X\sim\eta$ to refer to a random variable with distribution $\eta$. 
On a measurable space $(\varOmega,\mathcal{F})$, given a measurable function $\varphi:\varOmega\rightarrow\mathbb{R}$, a probability measure $\eta$, and a Markov transition kernel $M$, we define the integral $\eta(\varphi)=\int_{\varOmega}\varphi(x)\eta(dx)$ and the function $M(\varphi)(x)=\int_{\varOmega}\varphi(y)M(x,dy)$ for $x\in\varOmega$.

\subsection{Unbiased estimation with couplings \label{subsec:Context:-unbiased-estimation}}

Suppose $h:\mathbb{R}^d\rightarrow\mathbb{R}$ is a measurable function of interest and consider 
the task of approximating the integral $\pi(h)=\int h(x)\pi(dx)<\infty$. 
Following \citet{glynn2014exact} and \citet{jacob2017unbiased}, we will construct a pair of coupled Markov chains $X=(X_{n})_{n\geq0}$ and $Y=(Y_{n})_{n\geq0}$
with the same marginal law, associated with an initial distribution $\pi_{0}$ and a $\pi$-invariant Markov transition kernel $K$ defined on $\{\mathbb{R}^d, \mathcal{B}(\mathbb{R}^d)\}$. To do so, we introduce a Markov transition kernel 
$\bar{K}$ on $\{\mathbb{R}^d\times\mathbb{R}^d, \mathcal{B}(\mathbb{R}^d)\times\mathcal{B}(\mathbb{R}^d)\}$ that admits 
$K$ as its marginals, i.e. $\bar{K}\{(x,y),A\times\mathbb{R}^d\}=K(x,A)$ and $\bar{K}\{(x,y),\mathbb{R}^d\times A\}=K(y,A)$ 
for all $x,y\in\mathbb{R}^d$ and $A\in\mathcal{B}(\mathbb{R}^d)$. After initializing $(X_0,Y_0)\sim\bar{\pi}_0$ with a coupling that has $\pi_0$ as its marginals, we then simulate $X_1\sim K(X_0,\cdot)$ and $(X_{n+1},Y_n)\sim\bar{K}\{(X_n,Y_{n-1}),\cdot\}$ for all integer $n\geq 1$. We will write pr to denote the law of the coupled chain $(X_{n},Y_{n})_{n\geq0}$, and $E$ to denote
expectation with respect to pr. We now consider the following assumptions.

\begin{assumption}[Convergence of marginal chain]\label{ass:convergence}
As $n\to\infty$, we have $E\{h(X_{n})\}\to\pi(h)$. Furthermore, 
there exist $\kappa_1>0$ and $C_1<\infty$ such that 
$E\{h(X_{n})^{2+\kappa_1}\}<C_1$ for all integer $n\geq0$.
\end{assumption}
\begin{assumption}[Tail of meeting time]\label{ass:tail}
The meeting time $\tau=\inf\{ n\geq1:\;X_{n}=Y_{n-1}\}$
satisfies a geometric tail condition
of the form pr$(\tau>n)\leq C_2\kappa_2^{n}$ 
for some constants $C_2\in\mathbb{R}_+,\kappa_2\in(0,1)$ and all integer $n\geq0$.
\end{assumption}
\begin{assumption}[Faithfulness]\label{ass:faithfulness}
The coupled chains are faithful \citep{rosenthal1997faithful}, i.e. 
$X_{n}=Y_{n-1}$ for all integer $n\geq\tau$.
\end{assumption}

Under these assumptions, the random variable defined as
\begin{align}
H_{k}(X,Y)=h(X_{k})+\sum_{n=k+1}^{\tau-1}\left\{ h(X_{n})-h(Y_{n-1})\right\}\label{eq:Hk}
\end{align}
for any integer $k\geq0$, is an unbiased estimator of $\pi(h)$ with finite variance \citep[Proposition 3.1]{jacob2017unbiased}. 
Computation of (\ref{eq:Hk}) can be performed with $\tau-1$ applications of $\bar{K}$ and $\max(1,k+1-\tau)$ applications of $K$; thus the compute cost has a finite expectation under Assumption \ref{ass:tail}. 
The first term, $h(X_{k})$, is in general biased since the chain $(X_{n})_{n\geq0}$ might not have
reached stationarity by iteration $k$. The second term acts as a bias correction and is equal to zero when $k\geq\tau-1$. 

As the estimators $H_{k}(X,Y)$, for various values of $k$, can be computed 
from a single realization of the coupled chains, this prompts the definition of a 
time-averaged estimator $H_{k:m}(X,Y)=(m-k+1)^{-1}\sum_{n=k}^mH_n(X,Y)$ for integers $k\leq m$. The latter inherits the unbiasedness and finite variance properties, and can be rewritten as 
\begin{align}
H_{k:m}(X,Y)= M_{k:m}(X) + 
\sum_{n=k+1}^{\tau-1}\min\left(1,\frac{n-k}{m-k+1}\right)\left\{ h(X_{n})-h(Y_{n-1})\right\}\label{eq:Hkm}
\end{align}
where $M_{k:m}(X)=(m-k+1)^{-1}\sum_{n=k}^{m}h(X_{n})$ can be viewed as the usual Markov chain estimator with $m$ iterations and a burn-in period of $k-1$. 
As before, the second term plays the role of bias correction and is equal to zero when $k\geq\tau-1$. Hence if the value of $k$ is sufficiently large, we can expect the variance of 
$H_{k:m}(X,Y)$ to be close to that of $M_{k:m}(X)$.
Moreover, the cost of computing (\ref{eq:Hkm}), which involves $\tau-1$ applications of $\bar{K}$ and $\max(1,m+1-\tau)$ applications of $K$, becomes comparable to $m$ iterations under $K$ for sufficiently large $m$. 
Therefore we can expect the asymptotic inefficiency of $H_{k:m}(X,Y)$ in the limit of our computational budget, given by the product of the expected compute cost and the variance of $H_{k:m}(X,Y)$ \citep{glynn1992asymptotic}, to approach the asymptotic variance of the underlying Markov chain as $m$ increases. 
We refer to \citet[Section 3.1]{jacob2017unbiased} for a more detailed discussion on the impact of $k$ and $m$, and recall their proposed guideline of having $k$ as a large quantile of the meeting time $\tau$ and $m$ as a large multiple of $k$.

In practice, our proposed methodology involves simulating $R$ pairs of coupled Markov chains $(X^{(r)},Y^{(r)})=(X_{n}^{(r)},Y_{n}^{(r)})_{n\geq0},r=1,\ldots,R$ completely in parallel, with each pair taking a random compute time depending on their meeting time.
As this produces $R$ independent replicates $H_{k:m}(X^{(r)},Y^{(r)}),
r=1,\ldots,R$ of the unbiased estimator (\ref{eq:Hkm}), 
one can compute the average $R^{-1}\sum_{r=1}^RH_{k:m}(X^{(r)},Y^{(r)})$ to approximate $\pi(h)$. 
By appealing to the usual central limit theorem for independent and identically distributed random variables, 
confidence intervals that are justified as $R\rightarrow\infty$ 
can also be constructed. 

Explicit constructions of coupled chains satisfying Assumptions
\ref{ass:convergence}--\ref{ass:faithfulness} for Markov kernels $K$ that are
defined by Metropolis--Hastings algorithms and Gibbs samplers are given in
\citet[Section 4]{jacob2017unbiased} and \citet{jacob_smoothing2018}.  The focus of this article is to propose
a coupling strategy that is tailored for Hamiltonian Monte Carlo chains, so as
to enable the use of unbiased estimators (\ref{eq:Hk})--(\ref{eq:Hkm}).
We will illustrate in Section \ref{sec:Numerical-illustrations} that
this approach applies to realistic settings and retains 
the benefits of Hamiltonian Monte Carlo 
in terms of scaling with dimension.

\section{Hamiltonian dynamics \label{sec:Hamilton's-equations}}
\subsection{Hamiltonian flows \label{subsec:Hamiltonian-dynamics}}
Suppose that the target distribution has the form 
$\pi(dq)\propto\exp\{-U(q)\}dq$,
where the potential function $U:\mathbb{R}^{d}\rightarrow\mathbb{R}_{+}$ satisfies the following assumptions.
\begin{assumption}[Regularity and growth of potential]\label{ass:potential}
The potential $U$ is twice continuously differentiable and its gradient $\nabla U:\mathbb{R}^d\rightarrow\mathbb{R}^d$ is globally $\beta$-Lipschitz, i.e. there exists $\beta>0$ such that $|\nabla U(q)-\nabla U(q')|\leq\beta|q-q'|$ for all $q,q'\in\mathbb{R}^{d}$. 
\end{assumption}
These assumptions imply at most quadratic growth of the potential, or equivalently 
that the tails of the target distribution are no lighter than Gaussian.

We now introduce Hamiltonian flows
on the phase space $\mathbb{R}^{d}\times\mathbb{R}^d$, which consists of position variables
$q\in\mathbb{R}^{d}$ and momentum variables $p\in\mathbb{R}^{d}$.
We will be concerned with a Hamiltonian function $\mathcal{E}:\mathbb{R}^{d}\times\mathbb{R}^{d}\rightarrow\mathbb{R}_{+}$
of the form 
$\mathcal{E}(q,p)=U(q)+|p|^{2}/2$.
We note the use of the identity mass matrix here and will rely on preconditioning in Section \ref{subsec:Cox-Process} to incorporate curvature properties of $\pi$. The time evolution of a particle $\{q(t),p(t)\}_{t\in\mathbb{R}_{+}}$
under Hamiltonian dynamics is described by the
ordinary differential equations 
\begin{align}\label{eq:hamilton_ode}
\frac{d}{dt}q(t) &= \nabla_{p}\mathcal{E}\{q(t),p(t)\} = p(t),\quad
\frac{d}{dt}p(t) = -\nabla_{q}\mathcal{E}\{q(t),p(t)\} = -\nabla U\{q(t)\}.
\end{align}
Under Assumption \ref{ass:potential}, (\ref{eq:hamilton_ode}) with
an initial condition $\{q(0),p(0)\}=(q_{0},p_{0})\in\mathbb{R}^{d}\times\mathbb{R}^{d}$
admits a unique solution globally on $\mathbb{R}_{+}$ \citep[p. 14]{lelievre2012}.
Therefore the flow map $\Phi_{t}(q_{0},p_{0})=\{q(t),p(t)\}$ is well-defined 
for any $t\in\mathbb{R}_{+}$, and we will write its projection onto the position
and momentum coordinates as $\Phi_{t}^{\circ}(q_{0},p_{0})=q(t)$ 
and $\Phi_{t}^{*}(q_{0},p_{0})=p(t)$ respectively. 

It is worth recalling that
Hamiltonian flows have the following properties.
\begin{property}[Reversibility]\label{property:reversibility}
For any $t\in\mathbb{R}_{+}$, the inverse flow map satisfies $\Phi_{t}^{-1}=M\circ\Phi_{t}\circ M$, where $M(q,p)=(q,-p)$ denotes momentum reversal.
\end{property}
\begin{property}[Energy conservation]\label{property:energy} The Hamiltonian function satisfies 
$\mathcal{E}\circ\Phi_{t}=\mathcal{E}$ for any $t\in\mathbb{R}_{+}$.
\end{property}
\begin{property}[Volume preservation]\label{property:volume}
For any $t\in\mathbb{R}_{+}$ and $A\in\mathcal{B}(\mathbb{R}^{2d})$, we have $\mathrm{Leb}_{2d}\{\Phi_{t}(A)\}=\mathrm{Leb}_{2d}(A)$, where $\mathrm{Leb}_{2d}$ denotes the Lebesgue measure on $\mathbb{R}^{2d}$. 
\end{property}
These properties imply that the extended target distribution on phase space
$\tilde{\pi}(dq,dp)\propto\exp\{-\mathcal{E}(q,p)\}dqdp$
is invariant under the Markov semi-group induced by the flow, i.e. 
for any $t\in\mathbb{R}_{+}$, the pushforward measure $\Phi_{t}\sharp\tilde{\pi}$, defined as $\Phi_{t}\sharp\tilde{\pi}(A)=\tilde{\pi}\{\Phi_{t}^{-1}(A)\}$ for $A\in\mathcal{B}(\mathbb{R}^{2d})$, is equal to $\tilde{\pi}$.

\subsection{Coupled Hamiltonian dynamics \label{subsec:Coupled-Hamiltonian-dynamics}}
We now consider the coupling of two particles $\{q^{i}(t),p^{i}(t)\}_{t\in\mathbb{R}_{+}},\ (i=1,2)$
evolving under (\ref{eq:hamilton_ode}) with initial conditions $\{q^{i}(0),p^{i}(0)\}=(q_{0}^{i},p_{0}^{i}),\ (i=1,2)$.
We first draw some insights from a Gaussian example. 
\begin{example}
Let $\pi$ be a Gaussian distribution on $\mathbb{R}$ with mean $\mu\in\mathbb{R}$
and variance $\sigma^{2}>0$. In this case, we have $U(q)=(q-\mu)^{2}/(2\sigma^{2}), \nabla U(q)=(q-\mu)/\sigma^{2}$ and the solution of (\ref{eq:hamilton_ode})
is 
\begin{align*}
\Phi_{t}(q_{0},p_{0})=\left(\begin{array}{c}
\mu+(q_{0}-\mu)\cos\left(\frac{t}{\sigma}\right)+\sigma p_{0}\sin\left(\frac{t}{\sigma}\right)\\
p_{0}\cos\left(\frac{t}{\sigma}\right)-\frac{1}{\sigma}(q_{0}-\mu)\sin\left(\frac{t}{\sigma}\right)
\end{array}\right).
\end{align*}
Hence the difference between particle positions is  
\begin{align*}
q^{1}(t)-q^{2}(t) &= (q_{0}^{1}-q_{0}^{2})\cos\left(\frac{t}{\sigma}\right)+\sigma(p_{0}^{1}-p_{0}^{2})\sin\left(\frac{t}{\sigma}\right).
\end{align*}
If we set $p_{0}^{1}=p_{0}^{2}$, then $|q^{1}(t)-q^{2}(t)|=|\cos(t/\sigma)|\,|q_{0}^{1}-q_{0}^{2}|$, so for any non-negative integer $n$, the particles meet exactly whenever $t=(2n+1)\pi\sigma/2$, and contraction occurs for any $t\neq\pi n\sigma$. 
\end{example}

This example motivates a coupling that simply assigns particles the
same initial momentum. Moreover, it also reveals that certain trajectory
lengths will result in larger contraction than others. We now examine
the utility of this approach more generally. Define $\Delta(t)=q^{1}(t)-q^{2}(t)$
as the difference between particle locations and note that 
\begin{align*}
\frac{1}{2}\frac{d}{dt}|\Delta(t)|^{2} = \Delta(t)^{\top}\left\lbrace p^{1}(t)-p^{2}(t)\right\rbrace.
\end{align*}
Therefore by imposing that $p^{1}(0)=p^{2}(0)$, the function $t\mapsto|\Delta(t)|$
admits a stationary point at time $t=0$. This is geometrically intuitive
as the trajectories at time zero are parallel to one another for an
infinitesimally small amount of time. To characterize this stationary
point, we compute 
\begin{align*}
\frac{1}{2}\frac{d^{2}}{dt^{2}}|\Delta(t)|^{2}=-\Delta(t)^{\top}\left[\nabla U\{q^{1}(t)\}-\nabla U\{q^{2}(t)\}\right] + |p^{1}(t)-p^{2}(t)|^{2} 
\end{align*}
and consider the following assumption. 
\begin{assumption}[Local convexity of potential]\label{ass:convexity}
There exists a compact set $S\in\mathcal{B}(\mathbb{R}^{d})$,
with positive Lebesgue measure, such that the restriction of $U$ to
$S$ is $\alpha$-strongly convex, i.e. there exists $\alpha>0$ such that 
$\left(q-q'\right)^{\top}\left\lbrace\nabla U(q)-\nabla U(q')\right\rbrace\geq\alpha|q-q'|^{2}$ 
for all $q,q'\in S$. 
\end{assumption}
Under Assumption \ref{ass:convexity}, we have 
\begin{align*}
\frac{1}{2}\frac{d^{2}}{dt^{2}}|\Delta(0)|^{2}\leq-\alpha|\Delta(0)|^{2}+|p^{1}(0)-p^{2}(0)|^{2} 
\end{align*}
if $q_{0}^{1},q_{0}^{2}\in S$ and $q_{0}^{1}\neq q_{0}^{2}$. 
Therefore by taking $p^{1}(0)=p^{2}(0)$, it follows from the second derivative test that 
$t=0$ is a strict local maximum point.
Continuity of $t\mapsto|\Delta(t)|^{2}$ implies that there
exists a trajectory length $T>0$ such that for any $t\in(0,T]$, there exists 
$\rho\in[0,1)$ satisfying 
\begin{align}\label{eq:exact_contract_small_t}
|\Phi_{t}^{\circ}(q_{0}^{1},p_{0})-\Phi_{t}^{\circ}(q_{0}^{2},p_{0})|\leq\rho|q_{0}^{1}-q_{0}^{2}|.
\end{align}
We note the dependence of $T$ on the initial positions $q_{0}^{1}, q_{0}^{2}$ and momentum
$p_{0}$.  We now strengthen the above claim. 
\begin{lemma}
\label{lem:exact_contraction}Suppose that the potential $U$ satisfies Assumptions \ref{ass:potential}--\ref{ass:convexity}.  
For any compact set $A\subset S\times S\times\mathbb{R}^{d}$, there
exists a trajectory length $T>0$ such that for any $t\in(0,T]$, there exists 
$\rho\in[0,1)$ satisfying (\ref{eq:exact_contract_small_t}) for all $(q_{0}^{1},q_{0}^{2},p_{0})\in A$.
\end{lemma}
Although the qualitative result in Lemma \ref{lem:exact_contraction} is
sufficient for our purposes, we note that more quantitative results of this
type have been established recently by \citet[Theorem 6]{Mangoubi:2017} and
\citet[Theorem 2.1]{bourabee:2018} to study the mixing time of Hamiltonian
Monte Carlo.
The preceding results show that the trajectory length $T$ yielding contraction
of the coupled system and the corresponding contraction rate $\rho$ 
do not depend on $d$ but only on the constants
$\alpha$ and $\beta$ of Assumptions \ref{ass:potential}--\ref{ass:convexity}.  This suggests that such a coupling strategy can be
effective in high dimension as long as the Hessian of $U$ is sufficiently
well-conditioned. 

\section{Coupled Hamiltonian Monte Carlo\label{sec:Hamiltonian-Monte-Carlo}}

\subsection{Leap-frog integrator \label{subsec:Leap-frog-integrator}}
As the flow defined by (\ref{eq:hamilton_ode}) is typically intractable,
time discretizations are required. The leap-frog symplectic
integrator is a standard choice as it preserves Properties \ref{property:reversibility} and \ref{property:volume}.
Given a step size $\varepsilon>0$ and a number of leap-frog steps
$L\in\mathbb{N}$, this scheme initializes at $(q_{0},p_{0})\in\mathbb{R}^{d}\times\mathbb{R}^{d}$
and iterates 
\begin{align*}
p_{\ell+1/2} & =p_{\ell}-\frac{\varepsilon}{2}\nabla U(q_{\ell}),\quad
q_{\ell+1}     =q_{\ell}+\varepsilon p_{\ell+1/2}, \quad
p_{\ell+1}     =p_{\ell+1/2}-\frac{\varepsilon}{2}\nabla U(q_{\ell+1}),
\end{align*}
for $\ell=0,\ldots,L-1$. We write the leap-frog iteration as $\hat{\Phi}_{\varepsilon}(q_{\ell},p_{\ell})=(q_{\ell+1},p_{\ell+1})$
and the corresponding approximation of the flow as $\hat{\Phi}_{\varepsilon,\ell}(q_{0},p_{0})=(q_{\ell},p_{\ell})$
for $\ell=0,\ldots,L$. As before, we denote by $\hat{\Phi}_{\varepsilon,\ell}^{\circ}(q_{0},p_{0})=q_{\ell}$
and $\hat{\Phi}_{\varepsilon,\ell}^{*}(q_{0},p_{0})=p_{\ell}$ the
projections onto the position and momentum coordinates respectively.

It can be established that the leap-frog scheme is of order two \citep[Theorem 3.4]{hairer:2005}, 
i.e. for sufficiently small $\varepsilon$, we have  
\begin{align}
|\hat{\Phi}_{\varepsilon,L}(q_{0},p_{0})-\Phi_{\varepsilon L}(q_{0},p_{0})|& \leq C_{3}(q_0,p_0,L)\varepsilon^{2},\label{eq:leapfrog_traj_error} \\
|\mathcal{E}\{\hat{\Phi}_{\varepsilon,L}(q_{0},p_{0})\}-\mathcal{E}(q_{0},p_{0})| & \leq C_{4}(q_0,p_0,L)\varepsilon^{2},\label{eq:leapfrog_hamiltonian_error}
\end{align}
for some positive constants $C_{3}$ and $C_{4}$ that depend continuously on the initial condition $(q_0,p_0)$ for any number of leap-frog iterations $L$. 
To simplify our exposition and focus on the proposed methods,
we will assume throughout the article that
(\ref{eq:leapfrog_traj_error})--(\ref{eq:leapfrog_hamiltonian_error}) hold.
We refer to the book by \citet{hairer:2005} on geometric
numerical integration and to the survey by \citet{bou2018geometric} for
additional assumptions under which these error bounds hold. 

We now discuss how the above constants behave with dimension and integration length. 
Firstly, under the simplified setting of a target distribution with independent and identical marginals and appropriate growth conditions on 
the potential, the results of \citet[Proposition 5.3 \& 5.4]{beskos2013optimal} 
indicate that these constants would scale as ${d}^{1/2}$. 
Hence if we scale the step size $\varepsilon$ as $d^{-1/4}$, advocated by \citet{beskos2013optimal} in this setting,
we can expect these errors to be stable in high dimensions. 
Secondly, while the constant associated to the pathwise error bound (\ref{eq:leapfrog_traj_error}) will typically grow exponentially with $L$ \citep[Section 2.2.3]{Leimkuhler:2005}, 
the constant of the Hamiltonian error bound (\ref{eq:leapfrog_hamiltonian_error}) on the other
hand can be stable over exponentially long time intervals $\varepsilon L$
\citep[Theorem 8.1]{hairer:2005}. 
Although the Hamiltonian is not 
conserved exactly under time discretization, one can employ a Metropolis--Hastings
correction as described in the following section.

\subsection{Coupled Hamiltonian Monte Carlo kernel \label{subsec:Coupled-HMC-kernel}}

Hamiltonian Monte Carlo \citep{Duane:1987,Neal:1993} is a
Metropolis--Hastings algorithm that targets $\pi$ 
using time discretized Hamiltonian dynamics as proposals.
In view of Section \ref{subsec:Coupled-Hamiltonian-dynamics}, we
consider coupling two Hamiltonian Monte Carlo chains $(Q_{n}^{1},Q_{n}^{2})_{n\geq0}$ by initializing $(Q_0^1,Q_0^2) \sim \bar{\pi}_0$ and evolving the chains jointly according to the following procedure. 
\begin{algorithm}
    \caption{Coupled Hamiltonian Monte Carlo step given $(Q_{n-1}^1,Q_{n-1}^2)$.}
    \label{alg:coupled_hmc}
\begin{tabbing}
   Sample momentum $P_{n}^{*}\sim\mathcal{N}(0_{d},I_{d})$ and $U_n\sim\mathcal{U}[0,1]$ independently \\
   For $i=1,2$ \\
   \qquad Set $(q_{0}^i,p_{0}^i)=(Q_{n-1}^i,P_{n}^{*})$\\
   \qquad Perform leap-frog integration to obtain $(q_{L}^i,p_{L}^i)=\hat{\Phi}_{\varepsilon,L}(q_{0}^i,p_{0}^i)$ \\
   \qquad If $U_n < \alpha\{(q_{0}^i,p_{0}^i),(q_{L}^i,p_{L}^i)\}$, set $Q_{n}^i=q_{L}^i$ \\   
   \qquad Otherwise set $Q_{n}^i=Q_{n-1}^i$ \\
   Output $(Q_{n}^1,Q_{n}^2)$
\end{tabbing}
\end{algorithm}
Since the leap-frog integrator preserves Properties \ref{property:reversibility} and \ref{property:volume}, the Metropolis--Hastings acceptance probability is 
\begin{align}
\alpha\left\lbrace(q,p),(q',p')\right\rbrace=\min\left[1,\exp\left\lbrace \mathcal{E}(q,p)-\mathcal{E}(q',p')\right\rbrace\right],\label{eq:MH_acceptance}
\end{align}
for $(q,p),(q',p')\in\mathbb{R}^{d}\times\mathbb{R}^{d}$. 
Iterating the above yields two marginal chains $(Q_{n}^1)_{n\geq0}$ and $(Q_{n}^2)_{n\geq0}$ that are $\pi$-invariant.
Algorithm \ref{alg:coupled_hmc} amounts to running two Hamiltonian Monte Carlo chains with common
random numbers; this has been considered in \citet{neal2002circularly} to remove the burn-in bias, and in 
\citet{Mangoubi:2017} and \citet{bourabee:2018} to analyze mixing properties. 

We denote the associated coupled Markov transition kernel
on the position coordinates as $\bar{K}_{\varepsilon,L}\{(q^{1},q^{2}),A^{1}\times A^{2}\}$
for $q^{1},q^{2}\in\mathbb{R}^{d}$ and $A^{1},A^{2}\in\mathcal{B}(\mathbb{R}^{d})$.
Marginally we have $\bar{K}_{\varepsilon,L}\{(q^{1},q^{2}),A^{1}\times\mathbb{R}^{d}\}=K_{\varepsilon,L}(q^{1},A^{1})$
and $\bar{K}_{\varepsilon,L}\{(q^{1},q^{2}),\mathbb{R}^{d}\times A^{2}\}=K_{\varepsilon,L}(q^{2},A^{2})$,
where $K_{\varepsilon,L}$ denotes the Markov transition kernel of the marginal Hamiltonian Monte Carlo chain.
If we supplement Assumption \ref{ass:potential} with 
the existence of a local minimum of $U$, then
aperiodicity, Lebesgue-irreducibility and Harris recurrence of $K_{\varepsilon,L}$ follow from \citet[Theorem 2]{durmus2017convergence}; see also \citet{Cances:2007} and \citet{Livingstone:2016} for previous works. Hence ergodicity follows from \citet[Theorem 13.0.1]{meyn:tweedie:2009} and Assumption \ref{ass:convergence} is satisfied for test functions satisfying $\pi(h^{2+\kappa_1})<\infty$ for some $\kappa_1>0$. 

We will write the law of the coupled Hamiltonian Monte Carlo chain as pr$_{\varepsilon,L}$, and ${E}_{\varepsilon,L}$ to denote expectation with
respect to pr$_{\varepsilon,L}$. 
The following result establishes that the relaxed meeting time
$\tau_{\delta}=\inf\{ n\geq 0:|Q_{n}^{1}-Q_{n}^{2}|\leq\delta\}$, for any
$\delta>0$, has geometric tails.

\begin{theorem}\label{thm:relaxed_meeting}
Suppose that the potential $U$ satisfies Assumptions \ref{ass:potential}--\ref{ass:convexity}.  
Assume also that there exists $\tilde{\varepsilon}>0$ such that for any $\varepsilon\in(0,\tilde{\varepsilon})$ and $L\in\mathbb{N}$, there exist a measurable function $V:\mathbb{R}^d\rightarrow[1,\infty)$, $\lambda\in(0,1)$ and $b<\infty$ such that 
\begin{align}\label{eqn:drift}
K_{\varepsilon,L}(V)(q) \leq \lambda V(q) + b
\end{align}
for all $q\in\mathbb{R}^d$, $\pi_0(V)<\infty$ and 
$\{q\in \mathbb{R}^d: V(q)\leq \ell_1\}\subseteq\{q\in S : U(q)\leq \ell_0\}$
for some $\ell_0\in(\inf_{q\in S}U(q), \sup_{q\in S}U(q))$ and 
$\ell_1>1$ satisfying $\lambda + 2b(1-\lambda)^{-1}(1+\ell_1)^{-1} < 1$.
Then for any $\delta>0$, there exist $\varepsilon_0\in(0,\tilde{\varepsilon})$ and $L_0\in\mathbb{N}$ such that for any 
$\varepsilon\in(0,\varepsilon_0)$ and $L\in\mathbb{N}$ satisfying
$\varepsilon L<\varepsilon_0 L_0$, we have 
\begin{align}\label{eqn:relaxedmeeting_tails}	
	\mathrm{pr}_{\varepsilon,L}(\tau_{\delta}>n)\leq C_0\kappa_0^n
\end{align}
for some $C_0\in\mathbb{R}_+, \kappa_0\in(0,1)$ and all integer $n\geq 0$.
\end{theorem}

The proof of Theorem \ref{thm:relaxed_meeting} proceeds by first showing that the relaxed meeting can take place, in finite iterations, whenever both chains enter a region of the state space where the target distribution is strongly log-concave. As suggested in \citet{neal2002circularly}, one can expect good coupling behaviour if the chains spend enough time in this region of the state space; the second part of the proof makes this intuition precise by controlling 
excursions with the geometric drift condition (\ref{eqn:drift}). The latter can be established under additional assumptions on the potential $U$ \cite[Theorem 9]{durmus2017convergence}.

As Theorem \ref{thm:relaxed_meeting} implies that the coupled chains can get arbitrarily close with sufficient frequency, one could potentially employ the unbiased estimation framework of \citet{glynn2014exact} that 
introduces a truncation variable. 
To verify Assumption \ref{ass:tail} that requires exact meetings, 
in the next section, we combine the coupled Hamiltonian Monte Carlo kernel
with another coupled kernel that is designed to trigger exact meetings
when the two chains are close.

\section{Unbiased Hamiltonian Monte Carlo \label{sec:Proposed-estimators}}
\subsection{Coupled random walk Metropolis--Hastings kernel \label{subsec:Making-chains-meet}}
Let $K_{\sigma}$ denote the $\pi$-invariant Gaussian random walk Metropolis--Hastings kernel with proposal covariance $\sigma^2I_d$. 
The following describes a coupling of $K_{\sigma}(x,\cdot)$ and $K_{\sigma}(y,\cdot)$ that results in exact meetings with high probability when $x,y\in\mathbb{R}^d$ are close \citep{Johnson:1998,jacob2017unbiased} and $\sigma$ is appropriately chosen.

We begin by sampling the proposals $X^*\sim\mathcal{N}(x,\sigma^2I_d)$ and
$Y^*\sim\mathcal{N}(y,\sigma^2I_d)$ from the maximal coupling of these two
Gaussian distributions \citep[Section 4.1]{jacob2017unbiased}. Under the
maximal coupling, the probability of $\{X^*\neq Y^*\}$ is equal to the total variation
distance between the distributions
$\mathcal{N}(x,\sigma^{2}I_{d})$ and $\mathcal{N}(y,\sigma^{2}I_{d})$.
Analytical tractability in the Gaussian case allows us to write
that distance as
$\mathrm{pr}(2\sigma
|Z|\leq\delta)$, where $Z\sim\mathcal{N}(0,1)$ and $\delta=|x-y|$. By
approximating the
folded Gaussian cumulative distribution function \citep{pollard:2005}, we obtain 
\begin{align}\label{eqn:TV_approx}
\mathrm{pr}(X^{*}=Y^{*})=\mathrm{pr}(2\sigma |Z|>\delta)
=1-(2\pi)^{-1/2}\frac{\delta}{\sigma}
+\mathcal{O}\left(\frac{\delta^{2}}{\sigma^{2}}\right)
\end{align}
as $\delta/\sigma\rightarrow 0$. Hence to achieve pr$(X^{*}=Y^{*})=\theta$ for some desired probability
$\theta$, $\sigma$ should be chosen approximately as $\delta/\{(2\pi)^{1/2}\left(1-\theta\right)\}$. 

The proposed values $X^{*}$ and $Y^{*}$ are then accepted according to Metropolis--Hastings acceptance probabilities, i.e. if $U^*\leq\min\{1,\pi(X^{*})/\pi(x)\}$
and $U^*\leq\min\{1,\pi(Y^{*})/\pi(y)\}$ respectively, where
a common uniform random variable $U^*\sim\mathcal{U}[0,1]$ is used for both
chains. We denote the resulting coupled Markov transition kernel on $\{\mathbb{R}^d\times\mathbb{R}^d, \mathcal{B}(\mathbb{R}^d)\times\mathcal{B}(\mathbb{R}^d)\}$ as $\bar{K}_{\sigma}$. 
If $\sigma$ is small relative to the spread of the target distribution, 
the probability of accepting both proposals would be high.
On the other hand, (\ref{eqn:TV_approx}) shows that $\sigma$ needs to be 
large compared to $\delta$ for the event $\{X^{*}=Y^{*}\}$ to occur with high probability. 
This leads to a trade-off; in practice, 
one can monitor acceptance probabilities of random walk Metropolis--Hastings chains from 
preliminary runs to guide how small $\sigma$ should be.  
Although most simulations in Section \ref{sec:Numerical-illustrations} will employ $\sigma=10^{-3}$ as the default value, 
the sensitivity of the choice of $\sigma$ on our proposed methodology will be investigated in Sections \ref{subsec:Logistic-regression} and \ref{subsec:Cox-Process}. 

\subsection{Combining coupled kernels \label{subsec:Proposed-algorithm}}

We now combine the coupled Hamiltonian Monte Carlo kernel $\bar{K}_{\varepsilon,L}$ with 
the coupled random walk Metropolis--Hastings kernel $\bar{K}_{\sigma}$, 
introduced in Sections \ref{subsec:Coupled-HMC-kernel} and \ref{subsec:Making-chains-meet} respectively, using the following mixture 
\begin{align}\label{eqn:coupled_mixture}
\bar{K}_{\varepsilon,L,\sigma}\{(x,y),A\times B\} = (1-\gamma)\bar{K}_{\varepsilon,L}\{(x,y),A\times B\} + 
					\gamma\bar{K}_{\sigma}\{(x,y),A\times B\}
\end{align}
for $x,y\in\mathbb{R}^d$ and $A,B\in\mathcal{B}(\mathbb{R}^d)$, where $\gamma\in(0,1), \varepsilon>0, L\in\mathbb{N},\sigma>0$ are appropriately chosen. 
The rationale for this choice is to enable exact meetings using the coupled random walk Metropolis--Hastings kernel when the chains are brought close together by the coupled Hamiltonian Monte Carlo kernel. 

To address the choice of $\gamma$, in light of the efficiency considerations
in Section \ref{subsec:Context:-unbiased-estimation}, we should understand
how $\gamma$ impacts both the average meeting time, which we will investigate in
Sections \ref{subsec:Logistic-regression} and \ref{subsec:Cox-Process}, and the
asymptotic inefficiency of the marginal kernel
${K}_{\varepsilon,L,\sigma}=(1-\gamma)K_{\varepsilon,L}+\gamma K_{\sigma}$. 
We now compare the asymptotic inefficiency of ${K}_{\varepsilon,L,\sigma}$ to that of ${K}_{\varepsilon,L}$. 
Assuming that evaluation of the potential and its gradient have the same cost, the latter is given by the product of its cost $L+2$ and 
its asymptotic variance $v(h,K_{\varepsilon,L})=\lim_{n\rightarrow\infty}\mathrm{var}_{\varepsilon,L}\{n^{-1/2}\sum_{i=1}^nh(X_i)\}$ 
where $X_0\sim\pi$ and $X_n\sim K_{\varepsilon,L}(X_{n-1},\cdot)$ for all integer $n\geq 1$. Noting that the expected cost of 
${K}_{\varepsilon,L,\sigma}$ is $(1-\gamma)(L+2)+\gamma$, we now consider its asymptotic variance $v(h,K_{\varepsilon,L,\sigma})$. 
By Peskun's ordering \citep{peskun1973optimum}, we have $v(h,K_{\varepsilon,L,\sigma})\leq v(h,P_{\varepsilon,L})$ where 
$P_{\varepsilon,L}=(1-\gamma)K_{\varepsilon,L}+\gamma I$ 
with the identity kernel defined as $I(x,A)=\mathbb{I}_A(x)$ for $x\in\mathbb{R}^d$ and $A\in\mathcal{B}(\mathbb{R}^d)$. 
We then apply \citet[Corollary 1]{latuszynski2013clts} to obtain $v(h,K_{\varepsilon,L,\sigma})\leq \gamma(1-\gamma)^{-1}\mathrm{var}_{\pi}\{h(X)\}+(1-\gamma)^{-1}v(h,K_{\varepsilon,L})$. Hence in summary the relative asymptotic inefficiency can be upper bounded by 
\begin{align}\label{eqn:relative_compare_mixture}
\left\lbrace 1 + \gamma(1-\gamma)^{-1}(L+2)^{-1}\right\rbrace 
\left[ 1 + \gamma\{1+\Psi(h,K_{\varepsilon,L})\}^{-1}\right],
\end{align}
where $\Psi(h,K_{\varepsilon,L})=1+2\sum_{n=1}^{\infty}\mathrm{Corr}_{\varepsilon,L}\{h(X_0),h(X_n)\}$ denotes the 
integrated auto-correlation time of a stationary Hamiltonian Monte Carlo chain. In view of (\ref{eqn:relative_compare_mixture}), 
we advocate choosing only small values of $\gamma$ to reduce the loss of efficiency of the marginal chain; most simulations in Section \ref{sec:Numerical-illustrations} will employ $\gamma=1/20$ as the default value.

We will write $Q_{\sigma}(x,A)=\int_{A}\mathcal{N}(y;x,\sigma^2I_d)dy, x\in\mathbb{R}^d, A\in\mathcal{B}(\mathbb{R}^d)$ as 
the Markov transition kernel of the Gaussian random walk, the law of the resulting coupled chain $(X_n,Y_n)_{n\geq 0}$ as pr$_{\varepsilon,L,\sigma}$, and ${E}_{\varepsilon,L,\sigma}$ to denote expectation with
respect to pr$_{\varepsilon,L,\sigma}$. 
The following details the simulation of $(X_n,Y_n)_{n\geq 0}$ to compute the unbiased estimators described in Section \ref{subsec:Context:-unbiased-estimation}. 
\begin{algorithm}
    \caption{Compute unbiased estimator $H_{k:m}(X,Y)$ of $\pi(h)$ }
    \label{alg:coupled_mixture}
\begin{tabbing}
   Initialize $(X_0,Y_0)\sim\bar{\pi}_0$ from a coupling with $\pi_0$ as marginals \\
   With probability $\gamma$, sample $X_1\sim K_{\sigma}(X_0,\cdot)$; otherwise sample $X_1\sim K_{\varepsilon,L}(X_0,\cdot)$ \\
   Set $n=1$. While $n<\max(m,\tau)$\\
   \qquad With probability $\gamma$, sample $(X_{n+1},Y_n)\sim\bar{K}_{\sigma}\{(X_{n},Y_{n-1}),\cdot\}$\\
   \qquad Otherwise sample $(X_{n+1},Y_n)\sim\bar{K}_{\varepsilon,L}\{(X_{n},Y_{n-1}),\cdot\}$\\
   \qquad If $X_{n+1}=Y_n$ set $\tau = n+1$\\
   \qquad Increment $n \leftarrow n+1$\\
Compute $H_{k:m}(X,Y)$ using (\ref{eq:Hkm})
\end{tabbing}
\end{algorithm}

The mixture kernel $K_{\varepsilon,L,\sigma}$ inherits ergodicity properties from any of its components, therefore
Assumption \ref{ass:convergence} can be satisfied following the discussion in Section \ref{subsec:Coupled-HMC-kernel}.
Noting that the faithfulness property in Assumption
\ref{ass:faithfulness} holds by construction, we now turn our attention to Assumption \ref{ass:tail}. 
\begin{theorem}\label{thm:exact_meeting}
Suppose that the potential $U$ satisfies Assumptions \ref{ass:potential}--\ref{ass:convexity}.  
Assume also that there exist $\tilde{\varepsilon}>0$ and $\tilde{\sigma}>0$ such that for any $\varepsilon\in(0,\tilde{\varepsilon}),L\in\mathbb{N}$ and $\sigma\in(0,\tilde{\sigma})$, there exist a measurable function $V:\mathbb{R}^d\rightarrow[1,\infty)$, $\lambda\in(0,1), b<\infty$ and $\mu>0$ such that 
\begin{align}\label{eqn:new_drift}
K_{\varepsilon,L}(V)(x) \leq \lambda V(x) + b\quad\mbox{and}\quad Q_{\sigma}(V)(x)\leq \mu\{V(x)+1\}
\end{align}
for all $x\in\mathbb{R}^d$, $\pi_0(V)<\infty$, $\lambda_0=(1-\gamma)\lambda+\gamma(1+\mu)<1$  and 
$\{x\in \mathbb{R}^d: V(x)\leq \ell_1\}\subseteq\{x\in S : U(x)\leq \ell_0\}$
for some $\ell_0\in(\inf_{x\in S}U(x), \sup_{x\in S}U(x))$ and 
$\ell_1>1$ satisfying 
$\lambda_0 + 2\{(1-\gamma)b+\gamma\mu\}(1-\lambda_0)^{-1}(1+\ell_1)^{-1} < 1$.
Then there exist $\varepsilon_0\in(0,\tilde{\varepsilon}), L_0\in\mathbb{N}$ and $\sigma_0>0$ such that for any $\varepsilon\in(0,\varepsilon_0), L\in\mathbb{N}$ satisfying
$\varepsilon L<\varepsilon_0 L_0$ and $\sigma\in(0,\sigma_0)$, we have 
\begin{align}\label{eqn:exact_meetingtime}	
	\mathrm{pr}_{\varepsilon,L,\sigma}(\tau>n)\leq C_0\kappa_0^n
\end{align}
for some $C_0\in\mathbb{R}_+, \kappa_0\in(0,1)$ and all integer $n\geq 0$.
\end{theorem}

Proof of the above result proceeds in two parts as in Theorem
\ref{thm:relaxed_meeting}, but requires slightly stronger assumptions to ensure
that the mixture kernel still satisfies a geometric drift condition.  The
assumptions of Theorems \ref{thm:relaxed_meeting}--\ref{thm:exact_meeting} can
be verified for target distributions given by multivariate Gaussian
distributions and posterior distributions arising from Bayesian logistic
regression; see Section \ref{sec:check_assumptions} of the supplement.
Although the above discussion guarantees validity of the unbiased estimator
computed by Algorithm \ref{alg:coupled_mixture} for a range of tuning
parameters, its efficiency will depend on the distribution of the meeting time
$\tau$ induced by the coupling and mixing properties of the marginal kernel
$K_{\varepsilon,L,\sigma}$. 

\section{Numerical illustrations \label{sec:Numerical-illustrations}}

\subsection{Preliminaries}\label{sec:preliminaries}

In practice, we will run Algorithm \ref{alg:coupled_mixture} $R$ times independently in parallel to obtain the unbiased estimators 
$H_{k:m}(X^{(r)},Y^{(r)}), r=1,\ldots,R$.
Following the framework of \citet{glynn1992asymptotic}, we define the
asymptotic inefficiency in the limit of our computational budget as
$i(h,\bar{\pi}_0,\bar{K}_{\varepsilon,L,\sigma})=E_{\varepsilon,L,\sigma}\{2(\tau-1)+\max(1,m+1-\tau)\}\,\mathrm{var}_{\varepsilon,L,\sigma}\{H_{k:m}(X,Y)\}$,
assuming that applying $\bar{K}_{\varepsilon,L,\sigma}$ costs twice as much as
$K_{\varepsilon,L,\sigma}$. This measure of efficiency accounts for the fact that, with a given
compute budget, one can average over more estimators if each is cheaper to
compute. 
We will approximate this inefficiency by empirical averages over the $R$ realizations.
For comparison, the asymptotic variance $v(h,K_{\varepsilon,L})$
of the standard Hamiltonian Monte Carlo estimator 
will be approximated with the
\texttt{spectrum0.ar}
function of the \texttt{coda} R package \citep{codapackage} using $10,000$ iterations after a burn-in of $1,000$ for all examples. 
We will consider estimating first and second moments, i.e. set $h_{i}(x)=x_i$ and $h_{d+i}(x)=x_i^2$ for $i=1,\ldots,d$, and compare $i(\bar{\pi}_0,\bar{K}_{\varepsilon,L,\sigma})=\sum_{i=1}^{2d}i(h_{i},\bar{\pi}_0,\bar{K}_{\varepsilon,L,\sigma})$ with 
$v(K_{\varepsilon,L})=\sum_{i=1}^{2d}v(h_i,K_{\varepsilon,L})$ at possibly different parameter configurations. 
An important point to be illustrated in the following is that the parameters $\varepsilon$ and $L$ minimizing the asymptotic inefficiency $(L+2)v(K_{\varepsilon,L})$
might not necessarily be suitable for our proposed estimator. Lastly, we will employ the guideline of taking $k$ as the $90\%$ sample quantile of meeting times, obtained from a small number of preliminary runs, and setting $m = 10k$. 

\subsection{Toy examples \label{subsec:toy_example}}

We first investigate the scalability of the proposed approach in high
dimensions on a  standard Gaussian target distribution on $\mathbb{R}^d$, by
examining the average meeting time of stationary coupled chains generated by
(\ref{eqn:coupled_mixture}).  For simplicity, the parameters $\sigma=10^{-3}$
and $\gamma=1/20$ are taken as their default values.  To ensure stable
acceptance probabilities as $d\rightarrow\infty$ \citep{beskos2013optimal}, we
scale the step size as $\varepsilon=Cd^{-1/4}$ and select different constants
$C>0$ to induce a range of acceptance probabilities. The number of leap-frog
steps is taken as $L=1 + \lfloor \varepsilon^{-1} \rfloor$, which fixes the
integration time $\varepsilon L$ as approximately one.  For comparison, we
consider (\ref{eqn:coupled_mixture}) with $L=1$, as this corresponds to the
Metropolis-adjusted Langevin algorithm, and adopt the scaling
$\varepsilon^2=C^2d^{-1/3}$ \citep{roberts1998optimal}; see also Section
\ref{sec:coupling_mala} of supplementary material for an alternative coupling.
Lastly, we also consider coupled chains generated solely by the coupled random
walk Metropolis--Hastings kernel described in Section
\ref{subsec:Making-chains-meet}, with proposal variance scaled as
$\sigma^2=C^2d^{-1}$ \citep{roberts1997weak}.  The results displayed in Fig.
\ref{fig:gaussian} demonstrate the effectiveness of our coupling strategy in
high dimensions, and illustrates the appeal of Hamiltonian Monte Carlo 
kernels in high dimensional settings.

\begin{figure}
\begin{centering}
\begin{minipage}{0.3\textwidth}
\includegraphics[scale=0.35]{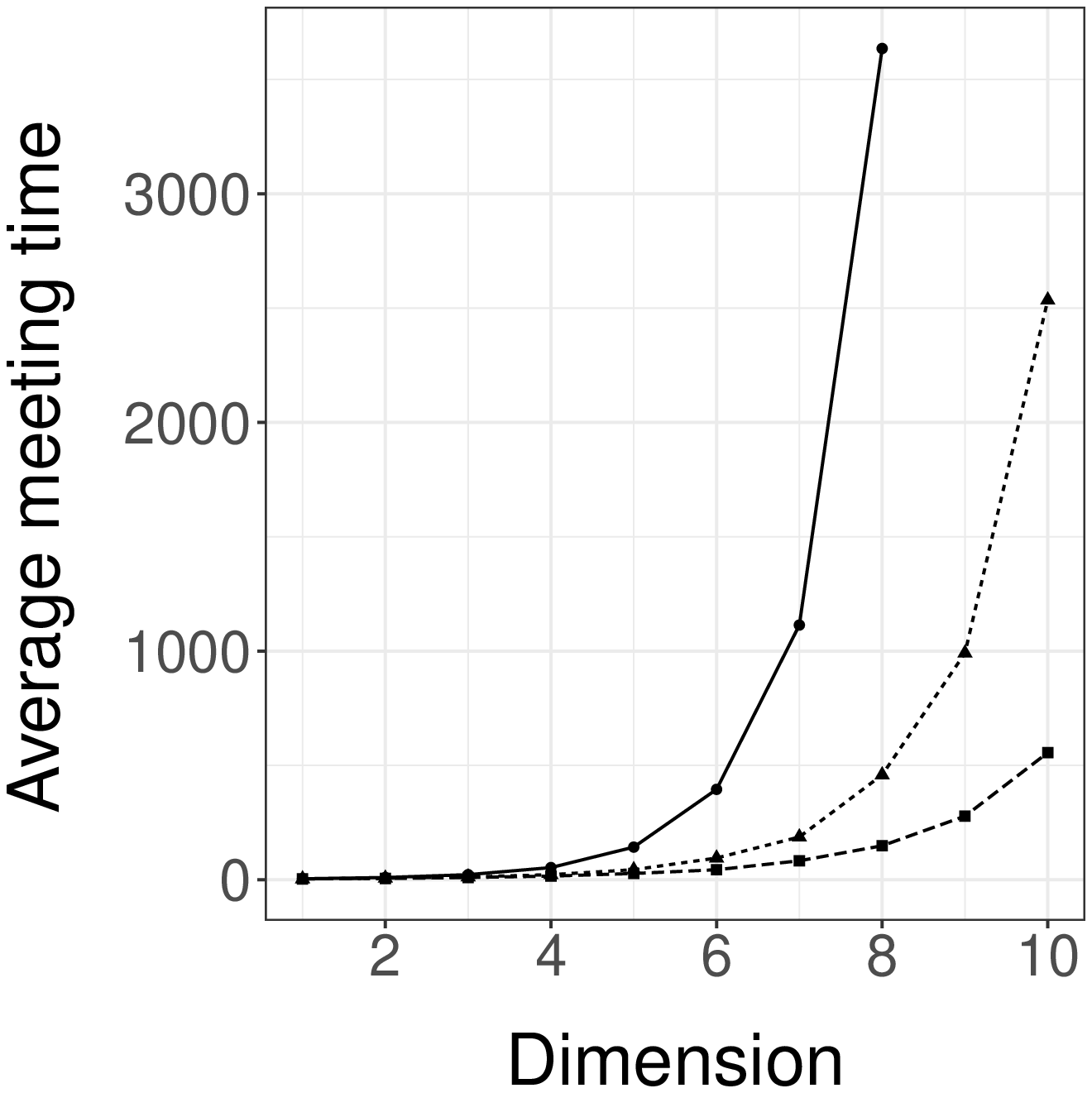}
\end{minipage}\hspace{0.5cm}
\begin{minipage}{0.3\textwidth}
\includegraphics[scale=0.35]{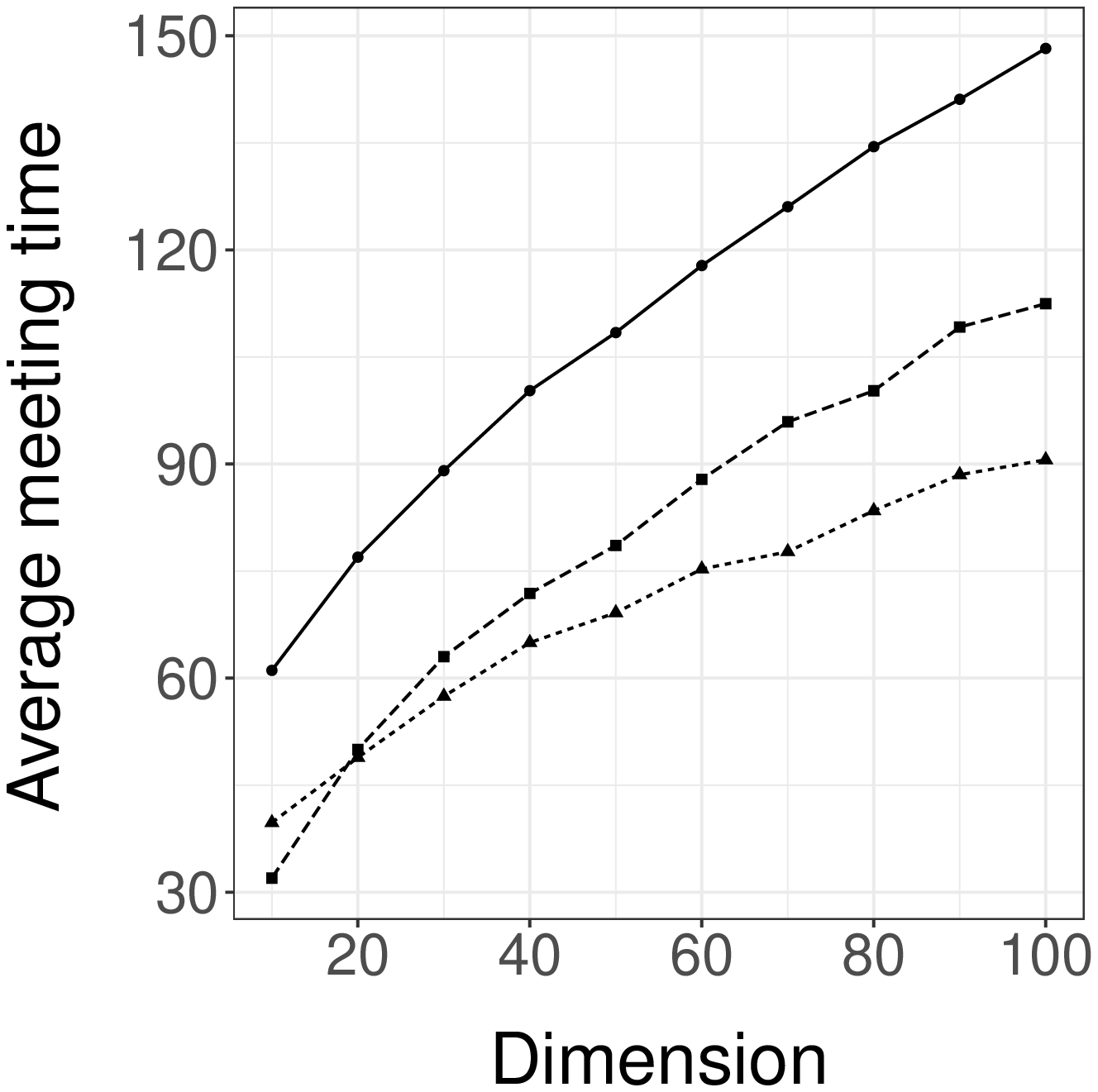}
\end{minipage}\hspace{0.5cm}
\begin{minipage}{0.3\textwidth}
\includegraphics[scale=0.35]{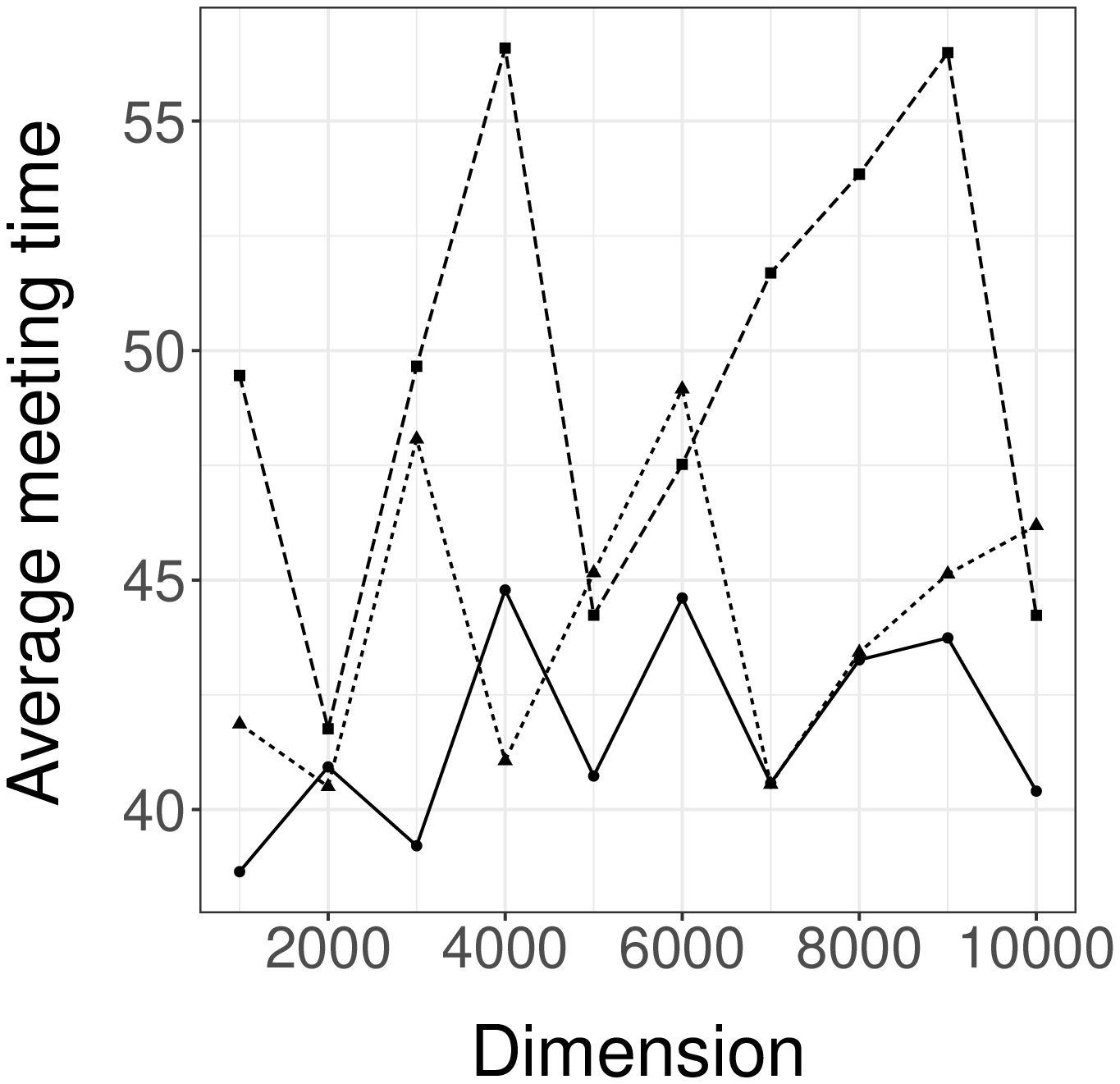}
\end{minipage}
\par\end{centering}
\caption{Gaussian example in Section \ref{subsec:toy_example}. Scaling of average meeting time with dimension for $1,000$ coupled chains based on random walk Metropolis--Hastings (left), Metropolis-adjusted Langevin algorithm (middle) and Hamiltonian Monte Carlo (right). 
The symbols and lines correspond to $C=1$ (dot-solid), $C=1.5$ (triangle-small dashes)
and $C=2$ (square-dashes).}
\label{fig:gaussian}
\end{figure}

Next we consider a banana-shaped target distribution on $\mathbb{R}^2$, whose potential is given by 
the Rosenbrock function $U(x_1,x_2)=(1-x_1)^2+10(x_2-x_1^2)^2$ for $(x_1,x_2)\in\mathbb{R}^2$. 
The aim here is to examine the utility of our proposed coupling for a
highly non-convex potential, 
and to explore the use of a new coupling for Hamiltonian Monte Carlo introduced
by \citet[Section 2.3.2]{bourabee:2018}. 
In contrast to Algorithm \ref{alg:coupled_hmc} which assigns the same initial
momentum to both chains, the latter samples an initial momentum
$P_n^1\sim\mathcal{N}(0_d,I_d)$ for the first chain, and sets the initial
momentum for the second chain as 
\begin{align*}
	P_n^2 = \begin{cases}
			P_n^1+\kappa\Delta_{n-1}, &\mbox{with probability }
			\frac{\mathcal{N}\left(\bar{\Delta}_{n-1}^{\top}P_n^1+\kappa|\Delta_{n-1}|;
			0,1\right)}{\mathcal{N}\left(\bar{\Delta}_{n-1}^{\top}P_n^1;0,1\right)},\\
			P_n^1-2(\bar{\Delta}_{n-1}^{\top}P_n^1)\bar{\Delta}_{n-1}, &\mbox{otherwise},
		      \end{cases} 
\end{align*}
where $\kappa>0$ is a tuning parameter, $\Delta_{n-1}=Q_{n-1}^1-Q_{n-1}^2$
denotes the difference between the chains at iteration $n-1$, and
$\bar{\Delta}_{n-1}=\Delta_{n-1}/|\Delta_{n-1}|$ the normalized difference.
Leap-frog integration and Metropolis--Hastings acceptance of the output are
then performed in the same way as Algorithm \ref{alg:coupled_hmc}; the
resulting coupled Hamiltonian Monte Carlo kernel is then employed in the
mixture (\ref{eqn:coupled_mixture}).  We simulate $1,000$ coupled chains,
initialized independently from the uniform distribution on $[-5,5]^2$, using
this new coupling with $\kappa=1$ and the previous one which corresponds 
to $\kappa=0$.  Employing the same parameters
$(\varepsilon,L,\sigma,\gamma)=(1/500, 500, 10^{-3}, 1/20)$ for both couplings,
we observe that the new coupling reduces the average meeting time from $158$ to
$52$. This example illustrates that the proposed methodology can be used beyond convex 
potentials, and that alternative couplings can result in significantly shorter meeting times.

\subsection{Logistic regression \label{subsec:Logistic-regression}}

We now consider a Bayesian logistic regression on the classic German credit
dataset, as in \citet{hoffman2014no}. After including all pairwise
interactions and performing standardization, the design matrix has $1,000$ rows
and $300$ columns.  Given covariates $x_i\in\mathbb{R}^{300}$, intercept
$a\in\mathbb{R}$ and coefficients $b\in\mathbb{R}^{300}$, each observation
$y_i\in\{0,1\}$ is modelled as an independent Bernoulli random variable with
probability of success $\{1+\exp(-a-b^{\top}x_i)\}^{-1}$. The prior is specified
as $a|s^2\sim\mathcal{N}(0,s^2), b|s^2\sim\mathcal{N}(0_{300},s^2I_{300})$
independently, and an Exponential distribution with rate $0.01$ for the
variance parameter $s^2$. 
The target $\pi$ is the posterior distribution of parameters $(a,b,\log s^2)$ on $\mathbb{R}^d$ with $d=302$. 

Initializing coupled chains independently from $\pi_0=\mathcal{N}(0_d,I_d)$,
for each parameter configuration
$(\varepsilon,L)\in\{0.01,0.0125,\ldots,0.04\}\times\{10,20,30\}$, we run $5$
pairs of coupled Hamiltonian Monte Carlo chains for $1,000$ iterations.  This
computation can be done independently in parallel for each configuration and
repeat; the output is displayed in the left panel of Fig. \ref{fig:logistic}.
Although multiple configurations lead to contractive
chains, it is not the case for 
$(\varepsilon,L)=(0.03,10)$ which are optimal parameters for Hamiltonian Monte Carlo. For configurations that yield distances that are
less than $10^{-10}$, we simulate $100$ meeting times in parallel using the
mixture kernel (\ref{eqn:coupled_mixture}) with 
$\sigma=10^{-3}$ and $\gamma=1/20$. We then select the parameter configuration
$(\varepsilon,L)=(0.0125,10)$ that gave the least average compute cost, taken
as $L+2$ times the average meeting time. 

To illustrate the impact of $\sigma$ and $\gamma$, 
we fix $(\varepsilon,L)=(0.0125,10)$ and examine the
distribution of meeting times as $\sigma$ or $\gamma$ varies. 
Decreasing $\sigma$ leads to larger meeting times: conservatively small values of $\sigma$ require
more iterations before the chains get close enough for the maximal coupling to
propose the same value with high probability. On the other hand,
if $\sigma$ was too large, 
large meeting times would be observed as 
random walk proposals would be rejected with high probability.
The middle panel of Fig. \ref{fig:logistic} suggests that the effectiveness of our coupling 
is not highly sensitive to the choice of $\sigma$, provided that it is small enough. 
Similarly, the right panel of Fig. \ref{fig:logistic} also shows stable meeting times 
for the range of values of $\gamma$ considered.

Finally, we produce $R=1,000$ coupled chains in parallel with
$(\varepsilon,L,\sigma,\gamma)=(0.0125,10,10^{-3},1/20)$ and compare the inefficiency of
our estimator with the asymptotic variance of the optimal Hamiltonian Monte
Carlo estimator for various choices of $k$ and $m$. The results, summarized in
Table \ref{table:logistic}, illustrate that bias removal comes at a cost
of increased variance, and that this can be reduced with appropriate choices of
$k$ and $m$.   Our guideline for $k$ and $m$ results in a relative inefficiency
of $1.05$ at an average compute cost of $3518$ applications of
$K_{\varepsilon,L,\sigma}$, or approximately $5$ minutes of computing time with
our implementation.
Therefore, thanks to unbiasedness, we can safely average over independent copies
of an estimator whose expected cost 
is of the order of a few thousand Hamiltonian Monte Carlo iterations.

\begin{figure}
\begin{centering}
\begin{minipage}{0.3\textwidth}
\includegraphics[scale=0.35]{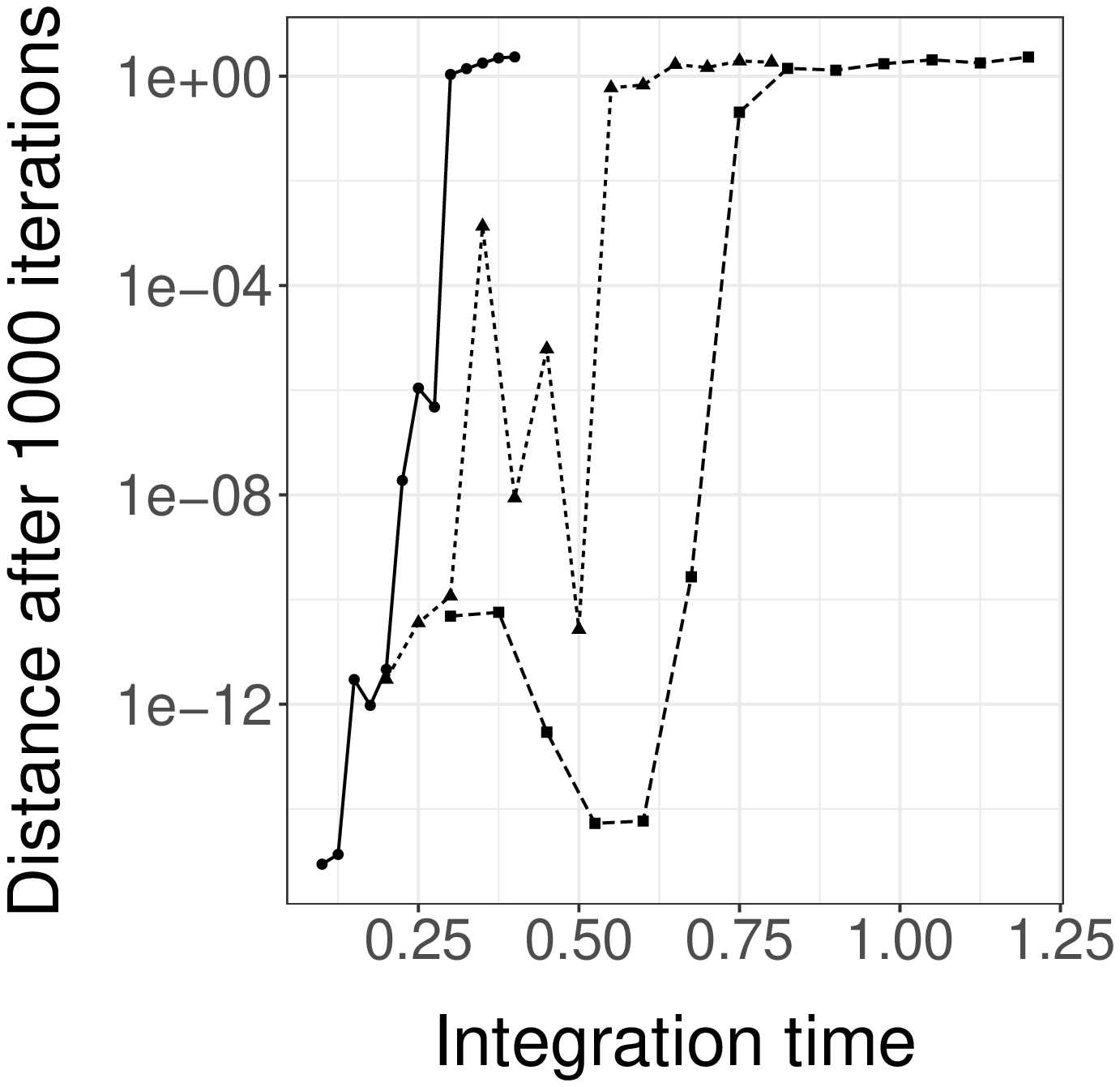}
\end{minipage}\hspace{0.5cm}
\begin{minipage}{0.3\textwidth}
\includegraphics[scale=0.35]{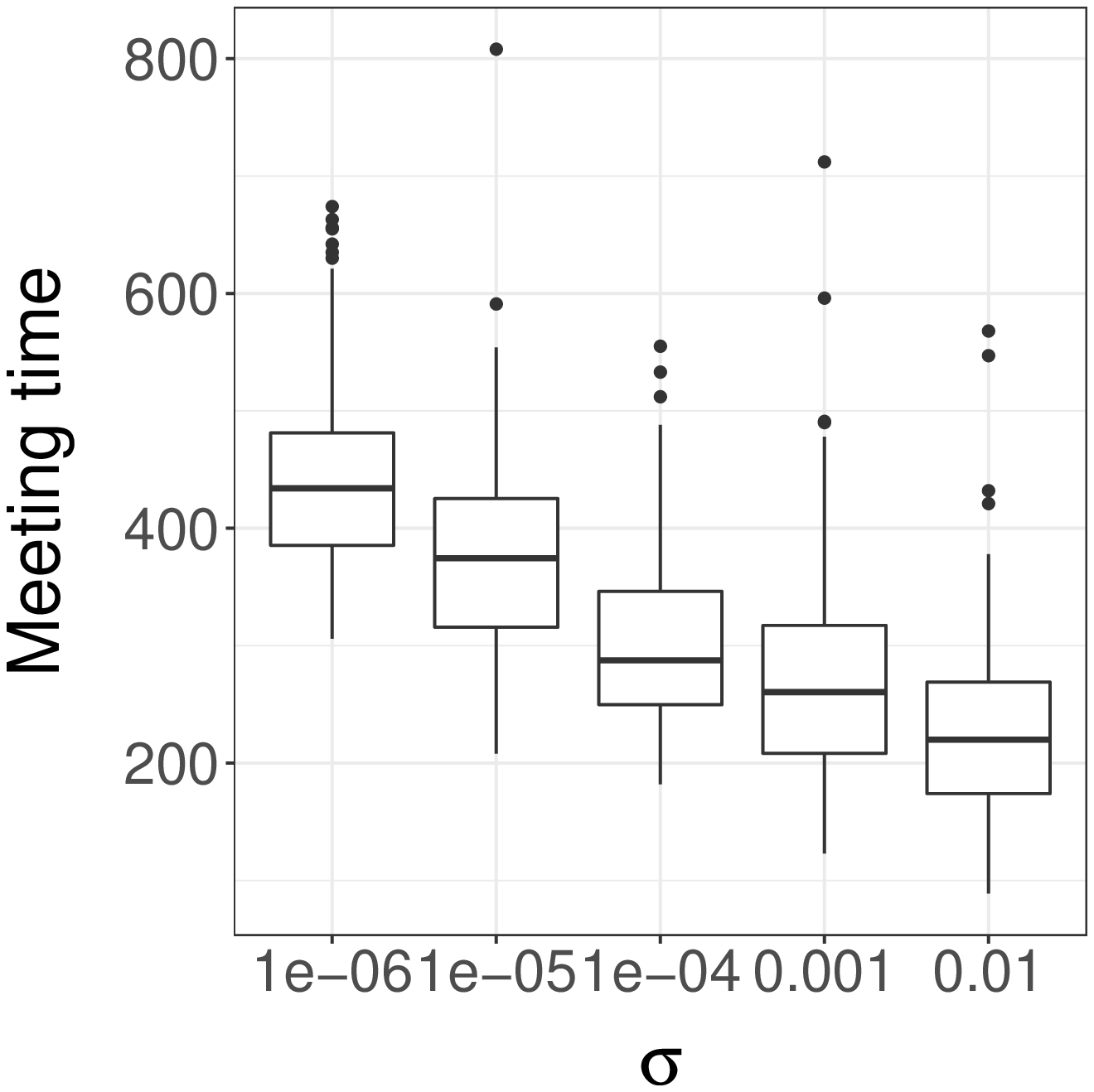}
\end{minipage}\hspace{0.5cm}
\begin{minipage}{0.3\textwidth}
\includegraphics[scale=0.35]{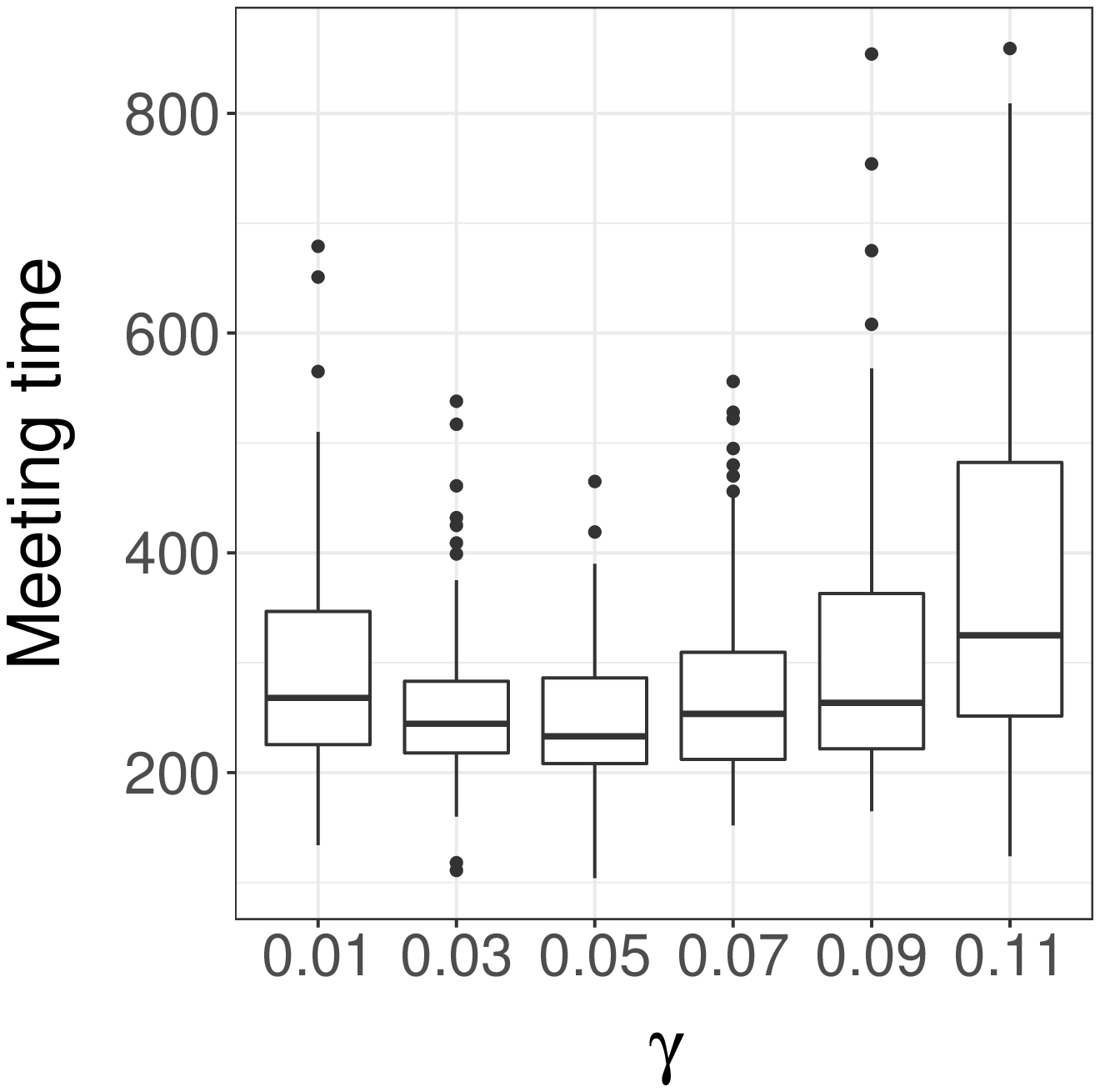}
\end{minipage}
\par\end{centering}
\caption{Logistic regression example in Section \ref{subsec:Logistic-regression}. Average distance between coupled chains at iteration $1,000$ against integration time $\varepsilon L$ (left). The symbols and lines correspond to $L=10$ (dot-solid), $L=20$ (triangle-small dashes)
and $L=30$ (square-dashes). Boxplot of meeting times as parameter $\sigma$ (middle) or $\gamma$ (right) varies.}
\label{fig:logistic}
\end{figure}

\begin{table}
\def~{\hphantom{0}}
\caption{Relative inefficiency of proposed estimator in logistic regression example\label{table:logistic}
}
\begin{center}
\begin{tabular}{ccccc}
$k$ & $m$ & Cost & Variance & Relative inefficiency \\[5pt]
$1$ & $k$ & $436$ & $4.0\times10^2$ & 1989.07 \\
$1$ & $5k$ & $436$ & $3.4\times10^2$ & 1671.93 \\
$1$ & $10k$ & $436$ & $2.8\times10^2$ & 1403.28 \\
$\mathrm{median}(\tau)$ & $k$ & 458 & $7.4\times10^0$ & 38.22 \\
$\mathrm{median}(\tau)$ & $5k$ & 1258 & ~$1.1\times10^{-1}$ & 1.58 \\
$\mathrm{median}(\tau)$ & $10k$ & 2298 & ~$4.5\times10^{-2}$ & 1.18 \\ 
$90\%\,\mathrm{quantile}(\tau)$ & $k$ & 553 & $6.0\times10^0$ & 38.11 \\ 
$90\%\,\mathrm{quantile}(\tau)$ & $5k$ & 1868 & ~$5.8\times10^{-2}$ & 1.23 \\ 
$90\%\,\mathrm{quantile}(\tau)$ & $10k$ & 3518 & ~$2.6\times10^{-2}$ & 1.05 \\ 
\end{tabular}
\caption*{Cost refers to the expected compute cost, variance denotes the sum of variances when estimating first and second moments, and relative inefficiency is the ratio of the asymptotic inefficiency $i(\bar{\pi}_0,\bar{K}_{\varepsilon,L,\sigma})$ with parameters $(\varepsilon,L,\sigma,\gamma)=(0.0125,10,10^{-3},1/20)$, to the asymptotic variance $v(K_{\varepsilon,L})$ with optimal parameters $(\varepsilon,L)=(0.03,10)$. These quantities were computed using $R=1,000$ independent runs, while the median and $90\%$ quantile of the meeting time were computed with $100$ preliminary runs.}
\end{center}
\end{table}


\subsection{Log-Gaussian Cox point processes \label{subsec:Cox-Process}}
We end with a challenging high dimensional application of 
Bayesian inference for log-Gaussian Cox point processes on a dataset
concerning the locations of $126$ Scot pine saplings in a natural forest in
Finland \citep{Moller:1998}. After discretizing the plot into an $n \times n$ regular grid, the number of points in each grid
cell $y_i\in\mathbb{N}$ is assumed to be conditionally independent, given a
latent intensity process $\Lambda_i, i\in\{1,\ldots,n\}^2$, and modelled as
Poisson distributed with mean $a\Lambda_i$, where $a=n^{-2}$ is the area of
each grid cell. The prior is specified by $\Lambda_i=\exp(X_i)$, where
$X_i,i\in\{1,\ldots,n\}^2$ is a Gaussian process with mean $\mu\in\mathbb{R}$
and exponential covariance function $\Sigma_{i,j}=s^2\exp\{-|i-j|/(nb)\}$ for
$i,j\in\{1,\ldots,n\}^2$. We will adopt the parameter values $s^2=1.91,b=1/33$
and $\mu=\log(126)-s^2/2$ estimated by \citet{Moller:1998} and infer the
posterior distribution of the latent process $X_i,i\in\{1,\ldots,n\}^2$ given
the count data and these hyperparameter values. 
We will consider three discretizations with $n\in\{16, 32, 64\}$, which 
correspond to target distributions $\pi$ on $\mathbb{R}^d$ with $d\in\{256, 1024, 4096\}$. 

Owing to the high dimensionality of this model, the mixing of random walk Metropolis--Hastings is known to be prohibitively slow \citep{Christensen:2002}, while the Metropolis-adjusted Langevin algorithm requires a computationally costly reparameterization to be effective \citep{Christensen:2005}. We will consider the use of Hamiltonian Monte Carlo and Riemann manifold Hamiltonian Monte Carlo with metric tensor $\Sigma^{-1}+a\exp(\mu+s^2/2)I_d$ \citep{girolami2011riemann}.
We proceed as in Section \ref{subsec:Logistic-regression} to seek parameter configurations 
$(\varepsilon,L)\in\{0.05,0.07,\ldots,0.45\}\times\{10,20,30\}$ that yield
contractive coupled chains with small compute cost, when initialized
independently from the prior distribution. Although both
algorithms have multiple configurations that result in contractive chains,
the parameters $\varepsilon$ and $L$ that were optimal for these methods
only led to contractive coupled Riemann manifold Hamiltonian Monte Carlo
chains for all three discretizations.  By simulating $100$ meeting times
with $\sigma=10^{-3}$ and $\gamma=1/20$ for configurations that yield
distances of less than $10^{-10}$, for $d\in\{256, 1024, 4096\}$
respectively, we select
$(\varepsilon,L)\in\{(0.11,10),(0.15,10),(0.17,10)\}$ for Hamiltonian Monte
Carlo, and $(\varepsilon,L)\in\{(0.11,10),(0.11,10),(0.13,10)\}$ for
Riemann manifold Hamiltonian Monte Carlo, which gave the smallest average
compute cost for each algorithm. The corresponding meeting times in the
left panel of Fig. \ref{fig:coxprocess} show the effectiveness of our
coupling strategy even in high dimensions.  The middle and right panels of
Fig. \ref{fig:coxprocess}, which display the meeting times of coupled
Riemann manifold Hamiltonian Monte Carlo chains for the finest
discretization, also illustrate the robustness of our coupling to the
choice of $\sigma$ and $\gamma$. 

With the above parameters and the guideline for choosing $k$ and $m$, we computed $R=1,000$ coupled chains in parallel for each algorithm and discretization. 
For $d\in\{256, 1024, 4096\}$ respectively, the relative inefficiency was found to be $11.00, 5.43, 2.73$ for 
Hamiltonian Monte Carlo, and $11.68, 7.85, 3.72$ for Riemann manifold Hamiltonian Monte Carlo. 
For the finest discretization, the average compute time was approximately $90$ and $20$ minutes with our implementation. Despite some loss of efficiency, 
the benefits of exploiting parallel computation for this problem is apparent since one can only run $4439$ and $714$ iterations of these algorithms respectively for the same compute time. 

\begin{figure}
\begin{centering}
\begin{minipage}{0.3\textwidth}
\includegraphics[scale=0.35]{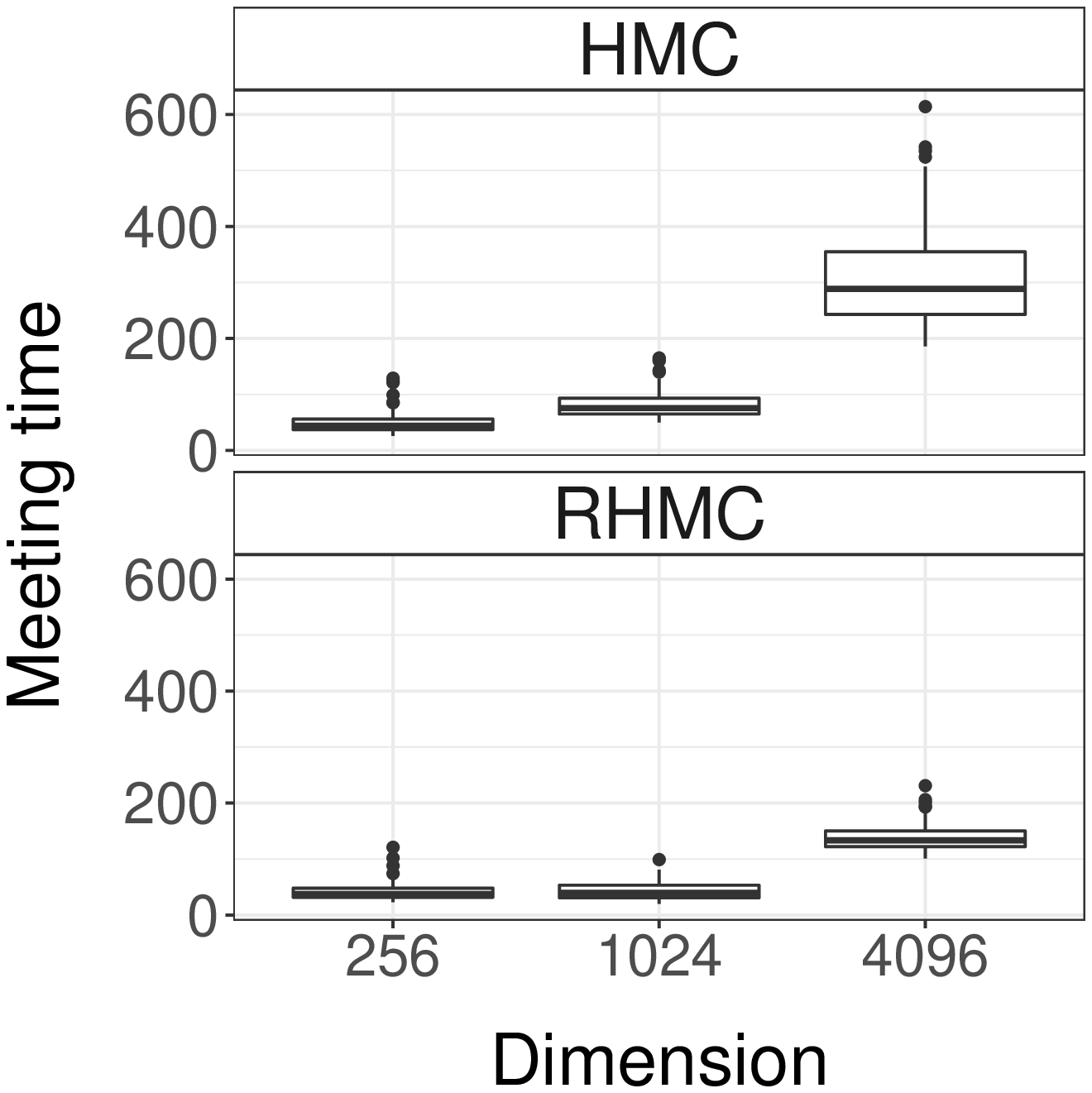}
\end{minipage}\hspace{0.5cm}
\begin{minipage}{0.3\textwidth}
\includegraphics[scale=0.35]{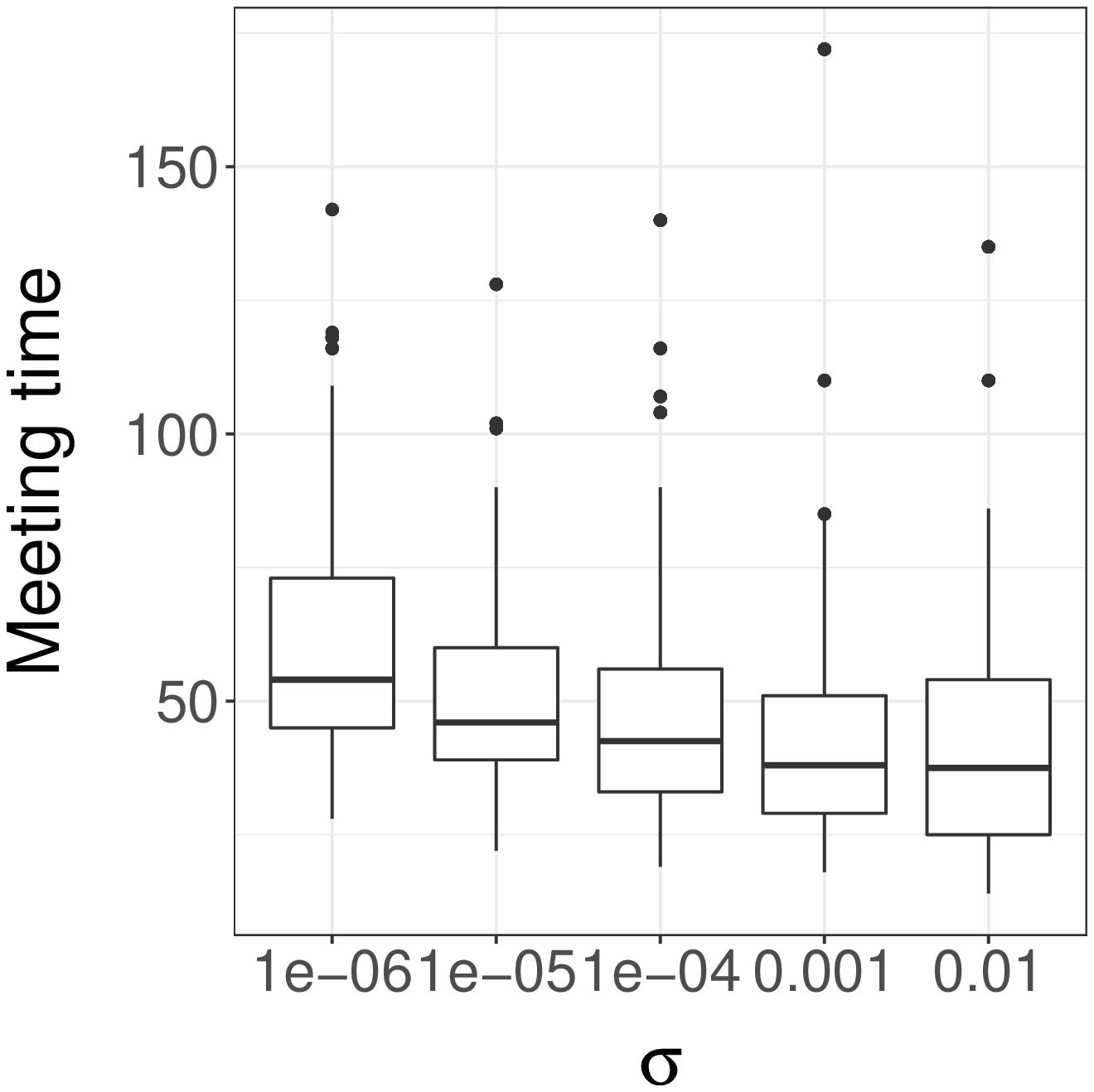}
\end{minipage}\hspace{0.5cm}
\begin{minipage}{0.3\textwidth}
\includegraphics[scale=0.35]{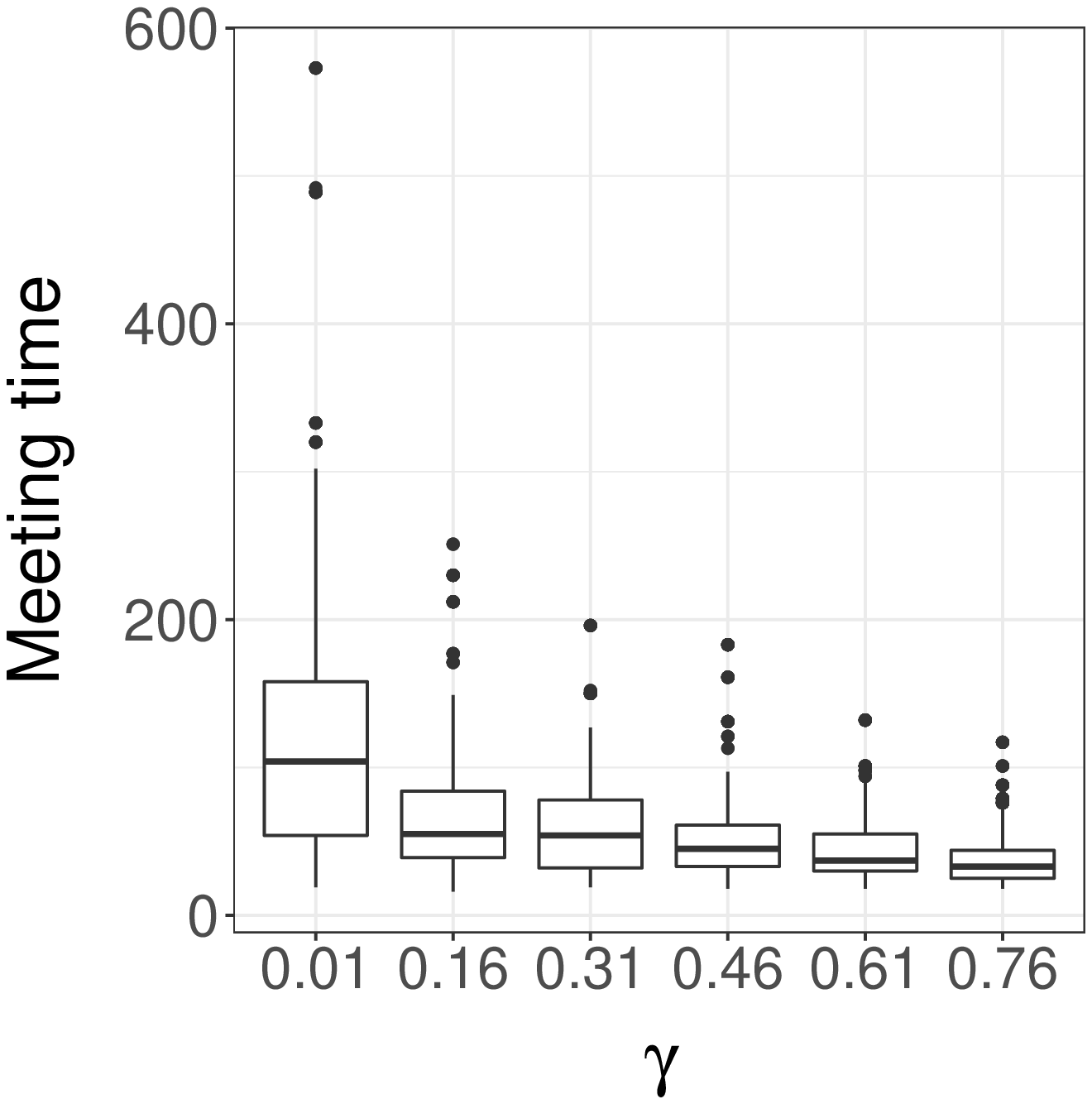}
\end{minipage}\hspace{0.5cm}
\par\end{centering}
\caption{Cox process example in Section \ref{subsec:Cox-Process}. Boxplot of meeting times for both algorithms and all three discretizations (left), and as parameter $\sigma$ (middle) or $\gamma$ (right) varies. }
\label{fig:coxprocess}
\end{figure}

\section{Discussion}

Construction of couplings could be explored for other variants of  
the Hamiltonian Monte Carlo method, such as the use of partial 
momentum refreshment \citep{horowitz1991generalized}, the adaptation 
of tuning parameters \citep{hoffman2014no}, different choices of 
kinetic energy \citep{livingstone2017kinetic}, and in combination with 
new sampling paradigms \citep{pollock2016scalable,fearnhead2016piecewise,vanetti2017piecewise}. 
Other ways of leveraging parallel hardware for Hamiltonian Monte Carlo
include the work in \citet{calderhead2014general}, which builds on \citet{tjelmeland2004using}
and focuses on parallel computation at each iteration of the algorithm. 
  
\section*{Acknowledgement}
The computations in this article were run on the Odyssey cluster supported by the FAS Division of Science, Research Computing Group at Harvard University.
Pierre E. Jacob gratefully acknowledges support by the National Science
Foundation through grant DMS-1712872. 
Both authors gratefully acknowledge support by the Army Research Office through grant W911NF-15-1-0172.

\section*{Supplementary material}
\label{SM}
An R package is available at
\href{https://github.com/pierrejacob/debiasedhmc}{github.com/pierrejacob/debiasedhmc}
and contains the scripts used to produce the figures of this article.
The supplementary material (available below) includes an alternative coupling 
for the Metropolis-adjusted Langevin algorithm, 
additional simulation results on truncated Gaussian distributions,
the proofs of Lemma \ref{lem:exact_contraction} and Theorems \ref{thm:relaxed_meeting}--\ref{thm:exact_meeting}, 
and notes on verifying the assumptions of Theorems \ref{thm:relaxed_meeting}--\ref{thm:exact_meeting} for 
target distributions given by posterior distributions of Bayesian logistic regression.

\bibliographystyle{abbrvnat}
\phantomsection\addcontentsline{toc}{section}{\refname}\bibliography{Biblio}

\appendix

\section{Coupling Metropolis-adjusted Langevin algorithm}\label{sec:coupling_mala}
We present an alternative to the construction in (\ref{eqn:coupled_mixture}) for the case $L=1$, which 
reduces to the Metropolis-adjusted Langevin algorithm with step size $\varepsilon^2>0$.
In this case, the coupled Hamiltonian Monte Carlo kernel $\bar{K}_{\varepsilon,1}$, introduced in Section \ref{sec:Hamiltonian-Monte-Carlo}, corresponds to a synchronous coupling of the proposal transition kernel $Q_{\varepsilon}(x,A)=\int_A\mathcal{N}\{y;x+\varepsilon^2\nabla\log\pi(x)/2,\varepsilon^2I_d\}dy, x\in\mathbb{R}^d, A\in\mathcal{B}(\mathbb{R}^d)$, associated to the Euler--Maruyama discretization of a $\pi$-invariant Langevin diffusion on $\mathbb{R}^d$ \citep{dalalyan_2017}. 

To construct a coupling of $K_{\varepsilon,1}(x,\cdot)$ and $K_{\varepsilon,1}(y,\cdot)$ that prompts exact meetings when $x,y\in\mathbb{R}^d$ are close, we can sample the proposals $(X^*,Y^*)$ from the maximal coupling of $Q_{\varepsilon}(x,\cdot)$ and $Q_{\varepsilon}(y,\cdot)$. 
Writing $\delta=|x-y|$, it follows from Assumption \ref{ass:potential} and the approximation in (\ref{eqn:TV_approx}) that 
\begin{align*}
	\mathrm{pr}(X^{*}=Y^{*})\geq1-\frac{(2+\beta\varepsilon^2)}{2(2\pi)^{1/2}}\frac{\delta}{\varepsilon}
+\mathcal{O}\left(\frac{\delta^{2}}{\varepsilon^{2}}\right)
\end{align*}
as $\delta/\varepsilon\rightarrow0$. 
As in Section \ref{subsec:Making-chains-meet}, the proposed values are then accepted with 
Metropolis--Hastings acceptance probabilities with a common uniform random variable for both chains. We denote the resulting coupled Markov transition kernel on $\{\mathbb{R}^d\times\mathbb{R}^d, \mathcal{B}(\mathbb{R}^d)\times\mathcal{B}(\mathbb{R}^d)\}$ as $\bar{K}_{\varepsilon}$. 
For some pre-specified threshold $\delta_0>0$, we can combine these coupled kernels by considering 
\begin{align*}
	\bar{K}\{(x,y),A\times B\} = 
	\mathbb{I}_{D_{\delta_0}^c}(x,y)\bar{K}_{\varepsilon,1}\{(x,y),A\times B\}
	+ \mathbb{I}_{D_{\delta_0}}(x,y)\bar{K}_{\varepsilon}\{(x,y),A\times B\}
\end{align*}
for $x,y\in\mathbb{R}^d$ and $A,B\in\mathcal{B}(\mathbb{R}^d)$, where $D_{\delta_0}=\{(x,y)\in\mathbb{R}^d\times\mathbb{R}^d : |x-y|\leq\delta_0\}$ and $D_{\delta_0}^c=\mathbb{R}^d\times\mathbb{R}^d\setminus D_{\delta_0}$. 
This coupled kernel admits the marginal Metropolis-adjusted Langevin algorithm kernel $K_{\varepsilon,1}$ as marginals, 
i.e. $\bar{K}\{(x,y),A\times\mathbb{R}^d\}=K_{\varepsilon,1}(x,A)$ and $\bar{K}\{(x,y),\mathbb{R}^d\times A\}=K_{\varepsilon,1}(y,A)$ 
for all $x,y\in\mathbb{R}^d$ and $A\in\mathcal{B}(\mathbb{R}^d)$, as this holds for 
both $\bar{K}_{\varepsilon,1}$ and $\bar{K}_{\varepsilon}$.

\section{Truncated Gaussian distribution \label{sec:Truncated-Normal-distribution}}
We investigate coupling Hamiltonian Monte Carlo on truncated Gaussian distributions constrained by quadratic inequalities. 
\citet{pakman2014exact} introduced an algorithm that generates 
trajectories which undergo exact Hamiltonian dynamics and bounce off the constraints. 
Implementing our method only involved simple modifications of their \texttt{tmg} R package
\citep{tmgpackage}. 

Following \citet{pakman2014exact}, we consider a bivariate standard Gaussian
distribution restricted to the set $\{(x_1,x_2)\in\mathbb{R}^2:
(x_1-4)^2/32+(x_2-1)^2/8\leq 1,
4x_{1}^{2}+8x_{2}^{2}-2x_{1}x_{2}+5x_{2}\geq1\}$ and use $\pi/2$ as the
trajectory length, as advocated in \citet{pakman2014exact}.  The left panel of Fig. \ref{fig:tmg}
displays $2,000$ Hamiltonian Monte Carlo samples.  Setting $(2,0)$ as the
initial position of both chains, the coupling proposed in Section
\ref{subsec:Coupled-Hamiltonian-dynamics} yields rapidly contracting chains
that are within machine precision in a few iterations.  A histogram of relaxed
meeting times, with respect to machine precision, is shown in the right panel
of Fig. \ref{fig:tmg}.  Our guideline for the choice of $k$ and $m$ yields
$k=6$ and $m=60$, based on $100$ draws of meeting times.  With these values, we
computed $R=1,000$ unbiased estimators and obtained an approximate asymptotic
inefficiency of $7.15$. In this case, the loss of efficiency is insignificant
compared to the Hamiltonian Monte Carlo algorithm, for which the asymptotic
variance was found to be approximately $6.41$.

\begin{figure}
\begin{centering}
\begin{minipage}{0.5\textwidth}
\includegraphics[scale=0.5]{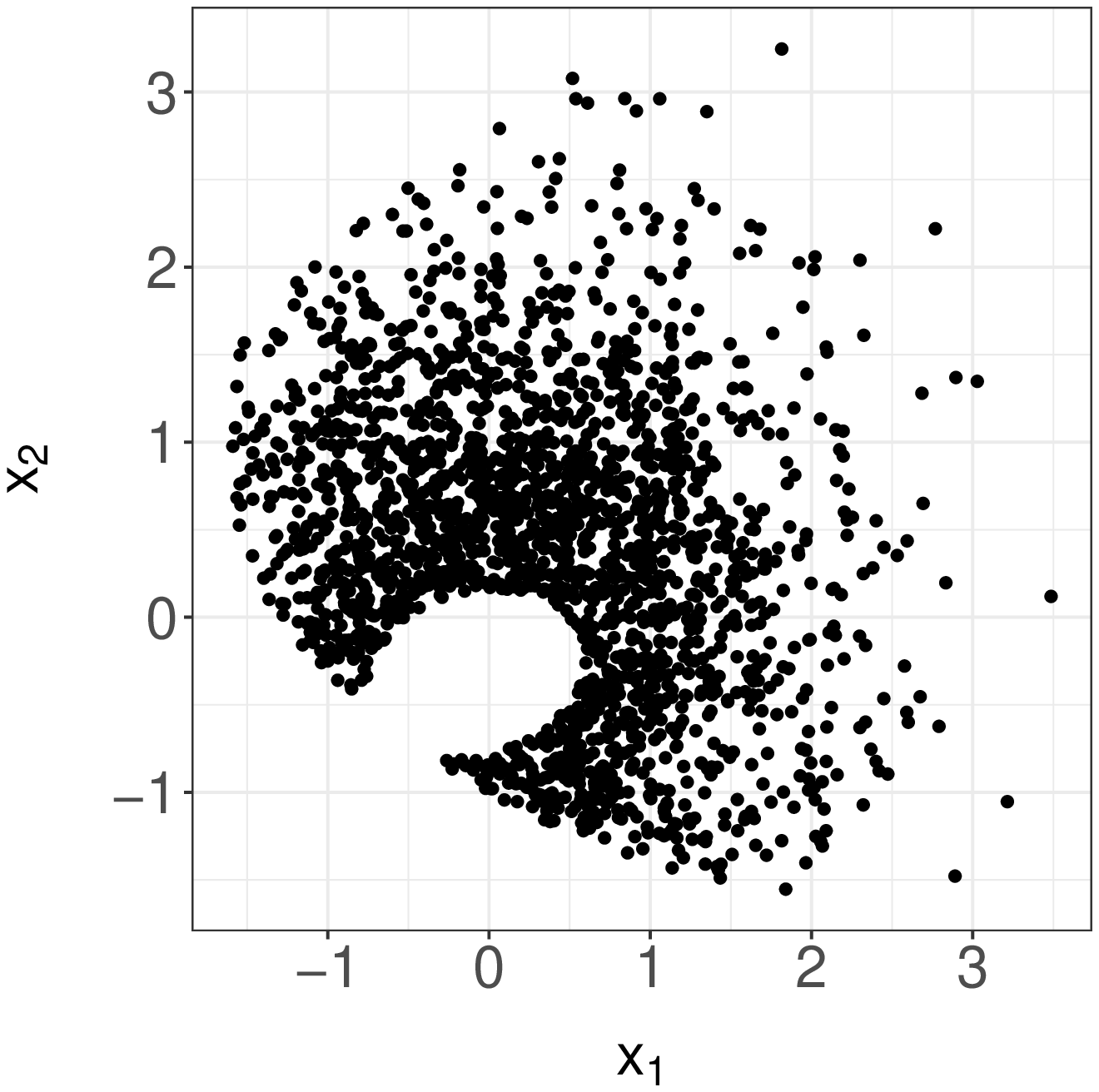}
\end{minipage}\hspace{-1cm}
\begin{minipage}{0.5\textwidth}
\includegraphics[scale=0.5]{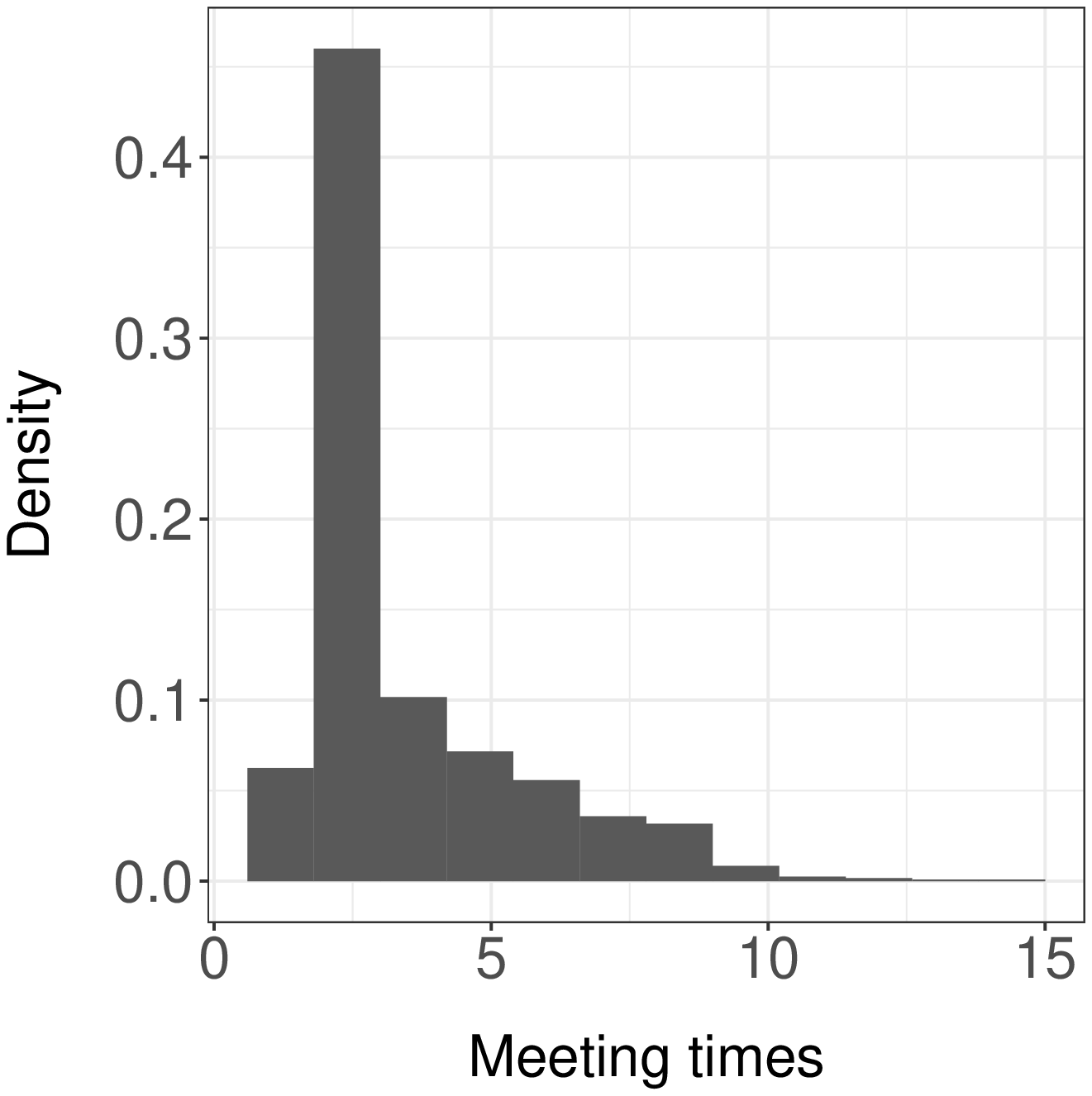}
\end{minipage}
\par\end{centering}
\caption{Truncated Gaussian example in Section \ref{sec:Truncated-Normal-distribution}. Scatter plot of $2,000$ Hamiltonian Monte Carlo samples approximating a Gaussian distribution truncated by quadratic constraints (left). Histogram of relaxed meeting
times for $1,000$ coupled Hamiltonian Monte Carlo chains (right).}
\label{fig:tmg}
\end{figure}

\section{Intermediate results}\label{sec:intermediate}
\begin{proof}[Proof of Lemma \ref{lem:exact_contraction}]
Take $(q_{0}^{1},q_{0}^{2},p_{0})\in A$. Applying Taylor's theorem on $\Delta(t)$ around $t=0$ gives 
\begin{align*}
\Delta(t)=\Delta(0)-\frac{1}{2}t^{2}G_{0}-\frac{1}{6}t^{3}G_{*}
\end{align*}
for some $t_{*}\in(0,t)$, where $G_{0} = \nabla U(q_{0}^{1})-\nabla U(q_{0}^{2})$
and 
\begin{align*}
G_{*} = \nabla^{2}U\{q^{1}(t_{*})\}p^{1}(t_{*})-\nabla^{2}U\{q^{2}(t_{*})\}p^{2}(t_{*}).
\end{align*}
We will control each term of the expansion 
\begin{align*}
|\Delta(t)|^{2}=|\Delta(0)|^{2}-t^{2}\Delta(0)^{\top}G_{0}-\frac{1}{3}t^{3}\Delta(0)^{\top}G_{*}+\frac{1}{4}t^{4}|G_{0}|^{2}+\frac{1}{6}t^{5}G_{0}^{\top}G_{*}+\frac{1}{36}t^{6}|G_{*}|^{2}.
\end{align*}
Using strong convexity, the Lipschitz assumption and Young's inequality
\begin{align*}
|\Delta(t)|^{2}\leq\left(1-\alpha t^{2}+\frac{1}{6}t^{3}+\frac{1}{4}\beta^{2}t^{4}+\frac{1}{12}\beta^{2}t^{5}\right)|\Delta(0)|^{2}+\left(\frac{1}{6}t^{3}+\frac{1}{12}t^{5}+\frac{1}{36}t^{6}\right)|G_{*}|^{2}.
\end{align*}
By Young's inequality and the Lipschitz assumption 
\begin{align*}
|G_{*}|^{2} & \leq2\|\nabla^{2}U\{q^{1}(t_{*})\}\|_{2}^{2}|p^{1}(t_{*})|^{2}+2\|\nabla^{2}U\{q^{2}(t_{*})\}\|_{2}^{2}|p^{2}(t_{*})|^{2}\\
 & \leq2\beta^{2}\left\lbrace|\Phi_{t_{*}}^{*}(q_{0}^{1},p_{0})|^{2}+|\Phi_{t_{*}}^{*}(q_{0}^{2},p_{0})|^{2}\right\rbrace\\
 & \leq2\beta^{2}\sup_{(q_{0}^{1},q_{0}^{2},p_{0})\in A}\left\lbrace|\Phi_{t_{*}}^{*}(q_{0}^{1},p_{0})|^{2}+|\Phi_{t_{*}}^{*}(q_{0}^{2},p_{0})|^{2}\right\rbrace
\end{align*}
where $\|\cdot\|_{2}$ denotes the spectral norm. The above supremum
is attained by continuity of the mapping $(q,p)\mapsto\Phi_{t_{*}}^{*}(q,p)$.
The claim (\ref{eq:exact_contract_small_t}) follows by combining
both inequalities and taking $t$ sufficiently small.
\end{proof}

As noted by an anonymous reviewer, inspection of the proof of Lemma \ref{lem:exact_contraction} reveals that one can relax 
local strong convexity in Assumption \ref{ass:convexity} to the condition 
\begin{align}\label{eqn:relax_convexity}
\left(q-q'\right)^{\top}\left\lbrace\nabla U(q)-\nabla U(q')\right\rbrace\geq f(|q-q'|)
\end{align}
for all $q,q'\in S$, where $f:\mathbb{R}_+\rightarrow\mathbb{R}_+$ is a function satisfying $f(x)>0$ whenever $x>0$ and 
$f(x)\geq Cx^2$ for some $C>0$ and all $x\in\mathbb{R}_+$. We now concern ourselves with an interpretation of 
(\ref{eqn:relax_convexity}). We shall assume in the following that $S$ contains a local mode, i.e. 
there exists $q^*\in S$ such that $\nabla U(q^*)=0$. It can be shown that strong convexity on $S$ is equivalent to 
\begin{align}
U(q) \geq U(q') + (q-q')^{\top}\nabla U(q')+\frac{\alpha}{2}|q-q'|^2
\end{align}
for all $q,q'\in S$. This implies $U(q)\geq U(q^*)+\alpha|q-q^*|^2/2$ for all $q\in S$, which can be seen as having the ratio of the target density and the Gaussian density $q\mapsto\mathcal{N}(q;q^*,\alpha^{-1}I_d)$ being upper bounded on $S$. 

Suppose additionally that $S$ is convex and $f$ is homogeneous of degree $k\in\mathbb{N}$, i.e. $f(cx)=c^kf(x)$ for all $c\in\mathbb{R}_+$ and $x\in\mathbb{R}_+$.
Fix $q,q'\in S$ and define the function $g(c)=U\{q'+c(q-q')\}$ for $c\in[0,1]$. We will write its derivative as 
$g'(c)=(q-q')^{\top}\nabla U\{q'+c(q-q')\}$. 
Applying (\ref{eqn:relax_convexity}) and homogeneity of $f$ gives $g'(c)\geq g'(0)+c^{k-1}f(|q-q'|)$. 
By continuity of $\nabla U$ and fundamental theorem of calculus 
\begin{align}
U(q)=g(1)=g(0)+\int_0^1g'(c)dc&\geq g(0) + g'(0) + f(|q-q'|)\int_0^1c^{k-1}dc \notag\\
&\geq U(q') + (q-q')^{\top}\nabla U(q') + \frac{1}{k}f(|q-q'|)\label{eqn:relaxed_first_order}
\end{align}
for all $q,q'\in S$. Therefore this implies that the ratio of the target density and the function $q\mapsto \exp\{-f(|q-q^*|)/k\}$ 
is upper bounded on $S$. Note also that (\ref{eqn:relaxed_first_order}) only implies $\left(q-q'\right)^{\top}\left\lbrace\nabla U(q)-\nabla U(q')\right\rbrace\geq (2/k)f(|q-q'|)$ for all $q,q'\in S$, so (\ref{eqn:relax_convexity}) and (\ref{eqn:relaxed_first_order}) are only equivalent when $k=2$.
This completes our discussion of (\ref{eqn:relax_convexity}).

To prove Theorem \ref{thm:relaxed_meeting}, we first establish the following intermediate result. 
For any measurable function $f:\varOmega\rightarrow\mathbb{R}$ and subset $A\subseteq \varOmega$,
we will write its level sets as $L_{\ell}(f)=\{x\in\varOmega:f(x)\leq\ell\}$ for $\ell\in\mathbb{R}$ and its restriction to $A$ as $f_A:A\rightarrow\mathbb{R}$. 

\begin{proposition}\label{prop:relaxed_meeting}
Suppose that the potential $U$ satisfies Assumptions \ref{ass:potential}--\ref{ass:convexity}.  
Then for any $\delta>0$, ${u}_0>\inf_{q\in S}U(q)$ and ${u}_1<\sup_{q\in S}U(q)$
with ${u}_0<{u}_1$, 
there exist $\bar{\varepsilon}>0$ and $\bar{L}\in\mathbb{N}$ such that for any 
$\varepsilon\in(0,\bar{\varepsilon})$ and $L\in\mathbb{N}$ satisfying
$\varepsilon L<\bar{\varepsilon} \bar{L}$, there exist $v_0 \in(u_0,u_1), n_0\in\mathbb{N}$ and $\omega\in(0,1)$
such that
\begin{align}\label{eqn:relaxed_meeting}
	\inf_{q^1,q^2\in S_0}\bar{K}_{\varepsilon,L}^{n_0}\{(q^1,q^2), {D}_{\delta}\} \geq \omega,
\end{align}
where $S_0 = L_{v_0}(U_S)$ is compact with positive Lebesgue measure, 
\begin{align*}
	\bar{K}_{\varepsilon,L}^{n}\{(q^1,q^2),A^1\times A^2\}=\mathrm{pr}_{\varepsilon,L}\{(Q_n^1,Q_n^2)\in A^1\times A^2\mid (Q_0^1,Q_0^2)=(q^1,q^2)\}
\end{align*}
denotes the $n$-step transition probabilities of the coupled chain, and ${D}_{\delta}=\{(q,q')\in\mathbb{R}^d\times\mathbb{R}^d: |q-q'|\leq \delta\}$.
\end{proposition}

\begin{proof}[Proof of Proposition \ref{prop:relaxed_meeting}]
Suppose that the chains $(Q_n^1)_{n\geq 0}, (Q_n^2)_{n\geq 0}$ are initialized at 
$Q_0^1=q^1\in S$ and $Q_0^2=q^2\in S$. Let $K(p)=|p|^{2}/2$ denote the kinetic energy function.
By compactness of $A=S\times S\times L_{k_0}(K)$, for some $k_0>0$ to be specified, it follows from Lemma \ref{lem:exact_contraction} that there exists a trajectory length $T>0$ such that for any $t\in(0,T]$, there exists 
$\rho_0\in[0,1)$ satisfying
\begin{align*}
|\Phi_{t}^{\circ}(Q_{0}^{1},P_{1}^*)-\Phi_{t}^{\circ}(Q_{0}^{2},P_{1}^*)|\leq\rho_0|Q_{0}^{1}-Q_{0}^{2}|
\end{align*}
for all $(Q_{0}^{1},Q_{0}^{2},P_{1}^*)\in A$.
Considering a fixed integration time $t\in(0,T]$, there exists $\omega_1\in(0,1)$ such that for any $\varepsilon>0$ and $L\in\mathbb{N}$
\begin{align*}
\mathrm{pr}_{\varepsilon,L}\left\lbrace |\Phi_{t}^{\circ}(Q_{0}^{1},P_{1}^*)-\Phi_{t}^{\circ}(Q_{0}^{2},P_{1}^*)|\leq\rho_0|Q_{0}^{1}-Q_{0}^{2}|\mid (Q_0^1,Q_0^2)=(q^1,q^2)\right\rbrace \geq \omega_1.
\end{align*}
By the triangle inequality, the pathwise error bound of the leap-frog integrator (\ref{eq:leapfrog_traj_error}) and compactness of $A$, there exist $\varepsilon_{1}>0$ and $\rho_{1}\in[0,1)$ such that 
\begin{align*}
\mathrm{pr}_{\varepsilon,L}\left\lbrace |\hat{\Phi}_{\varepsilon,L}^{\circ}(Q_{0}^{1},P_{1}^*)-\hat{\Phi}_{\varepsilon,L}^{\circ}(Q_{0}^{2},P_{1}^*)|\leq\rho_1|Q_{0}^{1}-Q_{0}^{2}|\mid (Q_0^1,Q_0^2)=(q^1,q^2)\right\rbrace \geq \omega_1
\end{align*}
for $\varepsilon\in(0,\varepsilon_{1})$ and $L\in\mathbb{N}$ satisfying $\varepsilon L=t$.
Using the Hamiltonian error bound of the leap-frog integrator  (\ref{eq:leapfrog_hamiltonian_error}) and compactness of $A$, it follows from 
(\ref{eq:MH_acceptance}) that there exist $\varepsilon_{2}\in(0,\varepsilon_{1}]$ and $\omega_2 \in (0,\omega_1)$ such that 
\begin{align*}
\mathrm{pr}_{\varepsilon,L}\left\lbrace Q_1^1=\hat{\Phi}_{\varepsilon,L}^{\circ}(Q_{0}^{1},P_{1}^*), Q_1^2=\hat{\Phi}_{\varepsilon,L}^{\circ}(Q_{0}^{2},P_{1}^*)\mid (Q_0^1,Q_0^2)=(q^1,q^2)\right\rbrace \geq 1-\omega_2
\end{align*}
for $\varepsilon\in(0,\varepsilon_{2})$ and $L\in\mathbb{N}$ satisfying
$\varepsilon L=t$. Noting that 
\begin{align*}
\left\lbrace |Q_1^1-Q_1^2|\leq\rho_1|Q_{0}^{1}-Q_{0}^{2}| \right\rbrace &\supseteq
\left\lbrace |\hat{\Phi}_{\varepsilon,L}^{\circ}(Q_{0}^{1},P_{1}^*)-\hat{\Phi}_{\varepsilon,L}^{\circ}(Q_{0}^{2},P_{1}^*)|\leq\rho_1|Q_{0}^{1}-Q_{0}^{2}| \right\rbrace \\
& \quad\cap \left\lbrace Q_1^1=\hat{\Phi}_{\varepsilon,L}^{\circ}(Q_{0}^{1},P_{1}^*), Q_1^2=\hat{\Phi}_{\varepsilon,L}^{\circ}(Q_{0}^{2},P_{1}^*)\right\rbrace,
\end{align*}
by Fr\'{e}chet's inequality 
\begin{align}\label{eqn:one_step_contraction}
\inf_{q^1,q^2\in S}\mathrm{pr}_{\varepsilon,L}\left\lbrace |Q_1^1-Q_1^2|\leq\rho_1|Q_{0}^{1}-Q_{0}^{2}|\mid (Q_0^1,Q_0^2)=(q^1,q^2)\right\rbrace \geq \omega_1 - \omega_2 > 0.
\end{align}

Consider $\delta>0$, ${u}_0>\inf_{q\in S}U(q)$ and ${u}_1<\sup_{q\in S}U(q)$ with ${u}_0<{u}_1$, and define
the sets $A_{\ell}=L_{\ell}(U_S) \times L_{u_1-\ell}(K)\subset L_{u_1}(\mathcal{E})$ for $\ell\in(u_0,u_1)$. 
As continuity and convexity of $U_S$ imply that it is a closed function, its level sets $L_{\ell}(U_S)$ for $\ell\in(u_0,u_1)$ are closed. Moreover, under the assumptions on $U$ and $S$, it follows that these level sets are compact with positive Lebesgue measure.
To iterate the argument in (\ref{eqn:one_step_contraction}), note first that if $(q,p)\in A_{\ell}$, Property \ref{property:energy} and continuity of $U$ and the mapping $t\mapsto\Phi_{t}^{\circ}(q,p)$ imply that $\Phi_{t}^{\circ}(q,p)\in L_{u_1}(U_S)$ for any $t\in\mathbb{R}_{+}$. 
Due to time discretization, we can only conclude using (\ref{eq:leapfrog_hamiltonian_error}) and compactness of $A_{\ell}$ that there exists $\eta_0>0$ such that
$\hat{\Phi}_{\varepsilon,L}^{\circ}(q,p)\in L_{u_1+\eta_0}(U)$ for all $(q,p)\in A_{\ell}$.
Set $n_0=\inf\{n\geq 1:\rho_1^{n}\sup_{q,q'\in S}|q-q'|\leq \delta\}$
and take $v_0\in(u_0,u_1), k_0>0,\eta_0>0$ small enough such that $v_0+(n_0+1)k_0+n_0\eta_0<u_1$ holds. Then we can conclude that 
\begin{align*}
\inf_{q^1,q^2\in L_{v_0}(U_S)}\mathrm{pr}_{\varepsilon,L}\left\lbrace |Q_{n_0}^1-Q_{n_0}^2|\leq\delta\mid (Q_0^1,Q_0^2)=(q^1,q^2)\right\rbrace \geq (\omega_1 - \omega_2)^{n_0} > 0
\end{align*}
and (\ref{eqn:relaxed_meeting}) follows.
\end{proof}

\section{Proofs of Theorems \ref{thm:relaxed_meeting} and \ref{thm:exact_meeting}}\label{sec:proofthm}
\begin{proof}[Proof of Theorem \ref{thm:relaxed_meeting}]	
For any $\delta>0$, we can apply Proposition \ref{prop:relaxed_meeting} with $u_0=\ell_0$ and any $u_1\in(\ell_0,\sup_{q\in S}U(q))$; the following adopts the notation in the conclusion of Proposition \ref{prop:relaxed_meeting}. This proof follows the arguments in \citet[Proposition 3.4]{jacob2017unbiased} with modifications to suit our setup. 
For $\varepsilon\in(0,\min\{\tilde{\varepsilon},\bar{\varepsilon}\})$ and $L\in\mathbb{N}$ satisfying $\varepsilon L<\bar{\varepsilon} \bar{L}$, it follows from assumption (\ref{eqn:drift}) that the coupled transition kernel $\bar{K}_{\varepsilon,L}$ satisfies the geometric drift condition 
\begin{align*}
	\bar{K}_{\varepsilon,L}(\bar{V})(q,q') \leq \lambda \bar{V}(q,q') + b
\end{align*}
for all $q,q'\in\mathbb{R}^d$ with $\bar{V}(q,q')=\{V(q)+V(q')\}/2$ as the bivariate Lyapunov function. Iterating gives 
$\bar{K}_{\varepsilon,L}^{n_0}(\bar{V})(q,q') \leq \lambda^{n_0} \bar{V}(q,q') + b/(1-\lambda)$.
For $(q,q')\notin L_{\ell_0}(U_S)\times L_{\ell_0}(U_S)$ which implies $(q,q')\notin L_{\ell_1}(V)\times L_{\ell_1}(V)$, we have $\bar{V}(q,q')\geq(1+\ell_1)/2$. Hence 
\begin{align}\label{eqn:niterate_drift}
	\bar{K}_{\varepsilon,L}^{n_0}(\bar{V})(q,q') \leq \lambda_0 \bar{V}(q,q')
\end{align}
with $\lambda_0=\lambda^{n_0}+2b(1-\lambda)^{-1}(1+\ell_1)^{-1}<1$ for all $(q,q')\notin L_{\ell_0}(U_S)\times L_{\ell_0}(U_S)$. Define the subsampled Markov chains $(\tilde{Q}_n^1)_{n\geq 0}, (\tilde{Q}_n^2)_{n\geq 0}$ as $\tilde{Q}_n^1=Q_{n_0n}^1, \tilde{Q}_n^2=Q_{n_0n}^2$ and the corresponding relaxed meeting time as $\tilde{\tau}_{\delta}=\inf\{n\geq 0:|\tilde{Q}_n^1-\tilde{Q}_n^2|\leq \delta\}$. For integers $n,j\geq 0$, consider the decomposition 
\begin{align}\label{eqn:decom}
	\mathrm{pr}_{\varepsilon,L}(\tilde{\tau}_{\delta}>n) = 
		\mathrm{pr}_{\varepsilon,L}(\tilde{\tau}_{\delta}>n, N_{n-1}\geq j) + 
		\mathrm{pr}_{\varepsilon,L}(\tilde{\tau}_{\delta}>n, N_{n-1}< j)
\end{align}
where $N_n$ denotes the number of times the coupled chain $(\tilde{Q}_k^1,\tilde{Q}_k^2)_{k\geq0}$ visits $L_{\ell_0}(U_S)\times L_{\ell_0}(U_S)$ by time $n$ (with $N_{-1}=0$). For the first term, it follows from (\ref{eqn:relaxed_meeting}) that 
\begin{align}\label{eqn:first_term}
	\mathrm{pr}_{\varepsilon,L}(\tilde{\tau}_{\delta}>n, N_{n-1}\geq j) \leq (1-\omega)^j.
\end{align}
To bound the second term, we define 
\begin{align}\label{eqn:B}
	B = \max\left\lbrace1, \frac{1}{\lambda_0}\sup_{(q,q')\in L_{\ell_0}(U_S)\times L_{\ell_0}(U_S)}\frac{\bar{K}_{\varepsilon,L}^{n_0}(\bar{V})(q,q')}{\bar{V}(q,q')}\right\rbrace
	\leq \frac{1}{\lambda_0}\left\lbrace \lambda^{n_0}+\frac{b}{1-\lambda}\right\rbrace
\end{align}
and apply Markov's inequality to obtain 
\begin{align}
	\mathrm{pr}_{\varepsilon,L}\left(\tilde{\tau}_{\delta}>n, N_{n-1}< j\right) &\leq 
	\mathrm{pr}_{\varepsilon,L}\left\lbrace\mathbb{I}_{D_{\delta}^c}(\tilde{Q}_n^1,\tilde{Q}_n^2)B^{-N_{n-1}}\geq B^{-(j-1)} \right\rbrace \notag\\
	&\leq B^{j-1}{E}_{\varepsilon,L}\left\lbrace \mathbb{I}_{D_{\delta}^c}(\tilde{Q}_n^1,\tilde{Q}_n^2)B^{-N_{n-1}}\right\rbrace \notag\\
	&\leq B^{j-1}{E}_{\varepsilon,L}\left\lbrace B^{-N_{n-1}}\bar{V}(\tilde{Q}_n^1,\tilde{Q}_n^2)\right\rbrace \notag\\
	&=\lambda_0^nB^{j-1}{E}_{\varepsilon,L}\left\lbrace M_n\right\rbrace\label{eqn:second_term}
\end{align}
where $M_n=\lambda_0^{-n}B^{-N_{n-1}}\bar{V}(\tilde{Q}_n^1,\tilde{Q}_n^2)$. 
Let $\mathcal{F}_n$ denote the $\sigma$-algebra generated by the random variables 
$(\tilde{Q}_k^1,\tilde{Q}_k^2)_{0\leq k\leq n}$.
We now establish that $(M_n, \mathcal{F}_n)_{n\geq 0}$ is a super-martingale. 
Suppose $(\tilde{Q}_n^1,\tilde{Q}_n^2)\notin L_{\ell_0}(U_S)\times L_{\ell_0}(U_S)$, in which case $N_n=N_{n-1}$, and 
applying (\ref{eqn:niterate_drift}) gives
\begin{align*}
	{E}_{\varepsilon,L}\left\lbrace M_{n+1}\mid\mathcal{F}_n\right\rbrace 
	&= \lambda_0^{-n-1}B^{-N_{n-1}}{E}_{\varepsilon,L}\left\lbrace \bar{V}(\tilde{Q}_{n+1}^1,\tilde{Q}_{n+1}^2)\mid 
	\tilde{Q}_n^1,\tilde{Q}_n^2\right\rbrace \leq M_n.
\end{align*}
For other case $(\tilde{Q}_n^1,\tilde{Q}_n^2)\in L_{\ell_0}(U_S)\times L_{\ell_0}(U_S)$, we have $N_n=N_{n-1}+1$  hence it follows from (\ref{eqn:B}) that
\begin{align*}
	{E}_{\varepsilon,L}\left\lbrace M_{n+1}\mid\mathcal{F}_n\right\rbrace 
	&= \lambda_0^{-n}B^{-N_{n-1}-1}\bar{V}(\tilde{Q}_n^1,\tilde{Q}_n^2)
	\frac{{E}_{\varepsilon,L}\left\lbrace \bar{V}(\tilde{Q}_{n+1}^1,\tilde{Q}_{n+1}^2)\mid\tilde{Q}_n^1,\tilde{Q}_n^2\right\rbrace}
	{\lambda_0\bar{V}(\tilde{Q}_n^1,\tilde{Q}_n^2)}
	\leq M_n.
\end{align*}
By the super-martingale property and assumption (\ref{eqn:drift}), ${E}_{\varepsilon,L}\{M_n\}\leq {E}_{\varepsilon,L}\{M_0\}\leq\{(\lambda+1)\pi_0(V)+b\}/2$. Therefore combining (\ref{eqn:decom}), (\ref{eqn:first_term}), (\ref{eqn:second_term}) and noting that the relaxed meeting times satisfy $\{\tau_{\delta}>n_0n\}\subseteq\{\tilde{\tau}_{\delta}>n\}$ give
\begin{align*}
	\mathrm{pr}_{\varepsilon,L}(\tau_{\delta}>n_0n)\leq \mathrm{pr}_{\varepsilon,L}(\tilde{\tau}_{\delta}>n) \leq (1-\omega)^j + 
	\frac{1}{2}\{(\lambda+1)\pi_0(V)+b\}\lambda_0^nB^{j-1}.
\end{align*}
Since $\lambda_0<1$, there exists $m_0\in\mathbb{N}$ such that $\lambda_0B^{1/m_0}<1$. For integer $n\geq m_0$, we can choose $j=\lceil n/m_0\rceil$ to obtain 
\begin{align*}
	\mathrm{pr}_{\varepsilon,L}(\tau_{\delta}>n_0n)\leq 
	\{(1-\omega)^{1/m_0}\}^n + 
	\frac{1}{2}\{(\lambda+1)\pi_0(V)+b\}(\lambda_0B^{1/m_0})^n
\end{align*}
which implies (\ref{eqn:relaxedmeeting_tails}).
\end{proof}

\begin{proof}[Proof of Theorem \ref{thm:exact_meeting}]
For any $\delta>0$, we can apply Proposition \ref{prop:relaxed_meeting} with $u_0=\ell_0$ and any $u_1\in(\ell_0,\sup_{q\in S}U(q))$; the following adopts the notation in the conclusion of Proposition \ref{prop:relaxed_meeting}.
Suppose that the coupled chain $(X_n,Y_n)_{n\geq 0}$ is initialized at 
$(X_0,Y_0)=(x,y)\in S_0\times S_0$ and evolves according to $(X_n,Y_n)\sim\bar{K}_{\varepsilon,L,\sigma}\{(X_{n-1},Y_{n-1}),\cdot\}$ for all integer $n\geq 1$, $\sigma>0$ and some $\varepsilon\in(0,\bar{\varepsilon}), L\in\mathbb{N}$ satisfying $\varepsilon L<\bar{\varepsilon}\bar{L}$ (note that this differs from the time shift presented in Algorithm \ref{alg:coupled_mixture}).
Let $\{I_n=1\}$ denote the event that the coupled Hamiltonian Monte Carlo kernel is sampled from the mixture (\ref{eqn:coupled_mixture}) at time $n$, i.e. $(I_n)_{n\geq1}$ is a sequence of independent Bernoulli random variables with probability of success $1-\gamma\in(0,1)$. By conditioning on the event $\cap_{n=1}^{n_0}\{I_n=1\}$, it follows from the proof of Proposition \ref{prop:relaxed_meeting} that there exist $n_0\in\mathbb{N}$ and $\omega\in(0,1)$ such that
\begin{align}\label{eqn:select_hmc}
	\inf_{x,y\in S_0}\mathrm{pr}_{\varepsilon,L,\sigma}\left\lbrace 
	(X_{n_0},Y_{n_0})\in D_{\delta}\cap S\times S \mid (X_0,Y_0)=(x,y)\right\rbrace \geq (1-\gamma)^{n_0}\omega.
\end{align}

Now conditioning on the events $\{(X_{n_0},Y_{n_0})\in D_{\delta}\cap S\times
S\}$ and $\{I_{n_0+1}=0\}$, for any $\sigma>0$ and $\theta_1\in(0,1)$, the
approximation (\ref{eqn:TV_approx}) allows us to select $\delta>0$ small enough
so that the maximal coupling within the coupled random walk
Metropolis--Hastings kernel $\bar{K}_{\sigma}$ proposes the same value
$X_{n_0+1}^*=Y_{n_0+1}^*$ with probability at least $1-\theta_1$. For any
$\theta_2\in (0,1)$, we now establish that the probability of accepting the proposed
value satisfies 
\begin{align}\label{eqn:controlling_acceptprob}
	\mathrm{pr}_{\varepsilon,L,\sigma}\left\lbrace X_{n_0+1}=X_{n_0+1}^* \mid 
	I_{n_0+1}=0, (X_{n_0},Y_{n_0})\in D_{\delta}\cap S\times S, (X_0,Y_0)=(x,y)\right\rbrace \geq 1-\theta_2
\end{align}
if $\sigma>0$ is sufficiently small. We can rewrite the above probability as   
\begin{align*}
	\mathrm{pr}_{\varepsilon,L,\sigma}\left\lbrace U_{n_0+1}\leq \min\left[1,\frac{\pi(X_{n_0}+\sigma Z_{n_0+1})}{\pi(X_{n_0})}\right] \mid 
	(X_{n_0},Y_{n_0})\in D_{\delta}\cap S\times S, (X_0,Y_0)=(x,y)\right\rbrace
\end{align*}
where $U_{n_0+1}\sim\mathcal{U}[0,1]$ and $Z_{n_0+1}=X_{n_0+1}^*/\sigma\sim\mathcal{N}(0,I_d)$ are independent.  
By Assumption \ref{ass:potential}, we have
\begin{align*}
\min\left[1,\frac{\pi(v+\sigma z)}{\pi(v)}\right]\geq\min\left[1, \exp\left\lbrace -\frac{1}{2}\sigma^2\beta|z|^2-\sigma\nabla U(v)^{\top}z\right\rbrace\right]
\end{align*}
for all $v,z\in\mathbb{R}^d$.  
Define $\varphi_1(\sigma,v,z)=\exp\{-\sigma^2\beta|z|^2/2\}$, $\varphi_2(\sigma,v,z)=\exp\{-\sigma\nabla U(v)^{\top}z\}$ and $B_0(r)=\{z\in\mathbb{R}^d : |z|\leq r\}$ for some $r>0$. Note that for each $(v,z)\in S\times B_0(r)$ and $i=1,2$, $\sigma\mapsto\varphi_i(\sigma,v,z)$ is a monotone function and $\lim_{\sigma\rightarrow0}\varphi_i(\sigma,v,z)=1$. Since $S\times B_0(r)$ is compact and $\nabla U$ is continuous, it follows from Dini's theorem that $\lim_{\sigma\rightarrow0}\inf_{v\in S, z\in B_0(r)}\varphi_i(\sigma,v,z)=1$. By conditioning on 
the events $\{X_{n_0}\in S\}$ and $\{Z_{n_0+1}\in B_0(r)\}$, we have 
\begin{align*}
	\left\lbrace U_{n_0+1}\leq \min\left[1,\frac{\pi(X_{n_0}+\sigma Z_{n_0+1})}{\pi(X_{n_0})}\right]\right\rbrace \supseteq \left\lbrace U_{n_0+1}\leq \min\left[1,\prod_{i=1}^2\inf_{v\in S, z\in B_0(r)}\varphi_i(\sigma,v,z)\right]\right\rbrace.
\end{align*}
The claim in (\ref{eqn:controlling_acceptprob}) follows by taking $r>0$ sufficiently large and $\sigma>0$ sufficiently small.
Therefore by symmetry of the coupled chains and Fr\'{e}chet's inequality, for any $\theta\in(0,1)$, there exists $\bar{\sigma}>0$ such that for any $\sigma\in(0,\bar{\sigma})$
\begin{align}\label{eqn:meet_with_RWMH}
	\mathrm{pr}_{\varepsilon,L,\sigma}\left\lbrace X_{n_0+1}=Y_{n_0+1} \mid 
	I_{n_0+1}=0, (X_{n_0},Y_{n_0})\in D_{\delta}\cap S\times S, (X_0,Y_0)=(x,y)\right\rbrace \geq 1-\theta.
\end{align}

Combining (\ref{eqn:select_hmc}) with (\ref{eqn:meet_with_RWMH}) gives 
\begin{align}\label{eqn:couldmeet_mixture}
	\inf_{x,y\in S_0}\bar{K}_{\varepsilon,L,\sigma}^{n_0+1}\{(x,y), {D}\} \geq (1-\gamma)^{n_0}\omega\gamma(1-\theta)>0
\end{align}
for $\varepsilon\in(0,\bar{\varepsilon}), L\in\mathbb{N}$ satisfying $\varepsilon L<\bar{\varepsilon}\bar{L}$ and $\sigma\in(0,\bar{\sigma})$, where $D=\{(x,y)\in\mathbb{R}^d\times\mathbb{R}^d : x=y\}$.
With (\ref{eqn:couldmeet_mixture}), the claim in (\ref{eqn:exact_meetingtime}) follows using the same arguments in the proof of Theorem \ref{thm:relaxed_meeting} since the marginal mixture kernel $K_{\varepsilon,L,\sigma}$ satisfies the geometric drift condition 
\begin{align*}
	K_{\varepsilon,L,\sigma}(V)(x) &= (1-\gamma)K_{\varepsilon,L}(V)(x)+\gamma K_{\sigma}(V)(x) \\
	&\leq (1-\gamma)\{\lambda V(x)+b\} + \gamma\{Q_{\sigma}(V)(x)+V(x)\} \\
	&\leq\lambda_0V(x) + b_0
\end{align*}
for all $x\in\mathbb{R}^d$ and $\sigma\in(0,\min\{\tilde{\sigma},\bar{\sigma}\})$, where $\lambda_0=(1-\gamma)\lambda+\gamma(1+\mu)\in(0,1)$ and $b_0=(1-\gamma)b+\gamma\mu<\infty$.

\end{proof}

\section{Verifying assumptions of Theorems \ref{thm:relaxed_meeting} and \ref{thm:exact_meeting} \label{sec:check_assumptions}}
\subsection{Model}
We consider the posterior distribution of regression coefficients $q\in\mathbb{R}^d$, arising from Bayesian logistic regression 
with observations $y\in\{0,1\}^N$ and 
a Gaussian prior distribution $\mathcal{N}(0,\zeta^{-1}\Sigma)$, where $\zeta>0$ controls the strength of the prior shrinkage toward zero. 
We will write the $n=1,\ldots,N$ row of the design matrix $X\in\mathbb{R}^{N\times d}$ as $x_n\in\mathbb{R}^d$. 

\subsection{Assumptions \ref{ass:potential}--\ref{ass:convexity}}
In the above setup, the potential has the form 
\begin{align*}
	U(q)=\frac{\zeta}{2}q^{\top}\Sigma^{-1} q + y^{\top}Xq + \sum_{n=1}^N\log\{1+\exp(-x_n^{\top}q)\} 
\end{align*}
which is infinitely differentiable. Its derivatives are given by 
\begin{align*}
	\nabla U(q)= \zeta\Sigma^{-1} q + X^{\top}y - \sum_{n=1}^N\frac{x_n}{1+\exp(x_n^{\top}q)} 
\end{align*}
and 
\begin{align*}
	\nabla^2U(q) = \zeta\Sigma^{-1} + \sum_{n=1}^N\frac{\exp(x_n^{\top}q)x_nx_n^{\top}}{\{1+\exp(x_n^{\top}q)\}^2}.
\end{align*}
The spectral norm of its Hessian can be bounded by 
\begin{align*}
\zeta\nu_{\min}(\Sigma^{-1})\leq \|\nabla^2U(q)\|_2\leq \nu_{\max}(\zeta\Sigma^{-1}+4^{-1}N\Sigma_X)
\end{align*}
for all $q\in\mathbb{R}^d$, where $\nu_{\min}(A)$ and $\nu_{\max}(A)$ denote the smallest and largest eigenvalues of a matrix $A\in\mathbb{R}^{d\times d}$ respectively, and $\Sigma_X=N^{-1}\sum_{n=1}^Nx_nx_n^{\top}$ is the Gram matrix. Therefore Assumption 
\ref{ass:potential} is satisfied with $\beta=\nu_{\max}(\zeta\Sigma^{-1}+4^{-1}N\Sigma_X)$ and Assumption \ref{ass:convexity} is 
satisfied on any compact set $S$ with $\alpha=\zeta\nu_{\min}(\Sigma^{-1})$. 
If we select $\Sigma=\Sigma_X^{-1}$, as considered in \citet[Example 2]{dalalyan2017theoretical}, then $\beta=(\zeta+N/4)\nu_{\max}(\Sigma_X)$ and $\alpha=\zeta\nu_{\min}(\Sigma_X)$.

\subsection{Geometric drift condition of Hamiltonian Monte Carlo kernel}
To establish that the marginal Hamiltonian Monte Carlo kernel satisfies a geometric drift condition (\ref{eqn:drift}), 
we will appeal to \citet[Theorem 9]{durmus2017convergence} which gives sufficient conditions \citet[Assumption H2($m$)]{durmus2017convergence} on the potential $U$ for geometric ergodicity. 
We will check the assumptions of \citet[Proposition 6]{durmus2017convergence} to verify \citet[Assumption H2($m$)]{durmus2017convergence}. To do so, we decompose the potential as $U(q)=U_0(q)+G(q)$ with 
\begin{align*}
	U_0(q) = \frac{\zeta}{2}q^{\top}\Sigma^{-1} q, \quad G(q) = y^{\top}Xq + \sum_{n=1}^N\log\{1+\exp(-x_n^{\top}q)\}.
\end{align*}

Firstly, $U_0$ and $G$ are infinitely differentiable. Secondly, $U_0$ satisfies $\lim_{|q|\rightarrow\infty}U_0(q)=\infty$, 
is homogeneous of degree $2$ and quasi-convex on 
$\mathbb{R}^d$; see discussion above \citet[Proposition 6]{durmus2017convergence} for precise definitions. 
Lastly, we need to show 
\begin{align}\label{eqn:tail_of_differential}
\lim_{|q|\rightarrow\infty}\|D^2G(q)\|=0\quad\mbox{ and }\quad
\lim_{|q|\rightarrow\infty}\|D^3G(q)\|\cdot|q|=0
\end{align}
where $D^k$ denotes the $k$ differential of $G$ and $\|D^kG\|$ is the operator norm of $D^k$ seen as a 
linear map from the $k$-fold product space $\mathbb{R}^d\times\cdots\times\mathbb{R}^d$ to $\mathbb{R}$. 
Let $|u|_{\infty}=\max_{i=1,\ldots,d}|u_i|$ denote the maximum norm for $u=(u_1,\ldots,u_d)\in\mathbb{R}^d$ 
and equip the product space with the norm $\|u\|_{k}=\max_{i=1,\ldots,k}|u_i|_{\infty}$ for 
$u=(u_1,\ldots,u_k)\in\mathbb{R}^d\times\cdots\times\mathbb{R}^d$. Note first that 
\begin{align}\label{eqn:partials_G}
\partial_i\partial_jG(q)=\sum_{n=1}^N\frac{\exp(x_n^{\top}q)x_{ni}x_{nj}}{\{1+\exp(x_n^{\top}q)\}^2}, \quad 
\partial_i\partial_j\partial_kG(q)= \sum_{n=1}^N\frac{\{\exp(x_n^{\top}q)-\exp(2x_n^{\top}q)\}x_{ni}x_{nj}x_{nk}}{\{1+\exp(x_n^{\top}q)\}^3},
\end{align}
where $\partial_if$ denotes the partial derivative of $f:\mathbb{R}^d\rightarrow\mathbb{R}$ with respect to 
the $i\in\{1,\ldots,d\}$ coordinate and $x_{ij}$ denotes the $(i,j)\in\{1,\ldots,d\}^2$ element of $X$. 
For $z=(u,v)\in\mathbb{R}^d\times\mathbb{R}^d$ with $u=(u_1,\ldots,u_d)$ and $v=(v_1,\ldots,v_d)$ in $\mathbb{R}^d$, 
we have 
\begin{align*}
\left|\sum_{i=1}^d\sum_{j=1}^du_iv_j\partial_i\partial_jG(q) \right|\leq |u|_{\infty}|v|_{\infty}
\sum_{i=1}^d\sum_{j=1}^d|\partial_i\partial_jG(q)|\leq \|z\|_2\sum_{i=1}^d\sum_{j=1}^d|\partial_i\partial_jG(q)|.
\end{align*}
Hence $\|D^2G(q)\|\leq\sum_{i=1}^d\sum_{j=1}^d|\partial_i\partial_jG(q)|$ and the same argument also gives 
$\|D^3G(q)\|\leq\sum_{i=1}^d\sum_{j=1}^d\sum_{k=1}^d|\partial_i\partial_j\partial_kG(q)|$. 
The claim (\ref{eqn:tail_of_differential}) then follows from the tail behaviour of (\ref{eqn:partials_G}).

Having established \citet[Assumption H2($m$)]{durmus2017convergence}, we apply 
\citet[Proposition 7]{durmus2017convergence} to conclude that the proposal Markov transition kernel with 
time discretized Hamiltonian dynamics, defined as 
\begin{align*}
	P_{\varepsilon,L}(q,A)=\int_{\mathbb{R}^d}\mathbb{I}_A\{\hat{\Phi}_{\varepsilon,L}^{\circ}(q,p)\}\mathcal{N}(p;0_d,I_d)dp
\end{align*}
for $q\in\mathbb{R}^d$ and $A\in\mathcal{B}(\mathbb{R}^d)$, satisfies a geometric drift condition, i.e. 
there exists $\tilde{\varepsilon}>0$ such that for any $\varepsilon\in(0,\tilde{\varepsilon})$ and $L\in\mathbb{N}$, 
there exist $a>0$, $\lambda_P\in(0,1)$ and $b_P>0$ such that 
\begin{align*}
	P_{\varepsilon,L}(V)(q)\leq\lambda_PV(q)+b_P
\end{align*}
for all $q\in\mathbb{R}^d$ with 
\begin{align}\label{eqn:explicit_lyapunov}
	V(q)=\exp(a|q|).
\end{align}
By \citet[Proposition 5]{durmus2017convergence} which also holds under \citet[Assumption H2($m$)]{durmus2017convergence}, 
the geometric drift condition for the proposal kernel implies a geometric drift condition for resulting 
Hamiltonian Monte Carlo kernel, i.e. for all $\varepsilon\in(0,\tilde{\varepsilon})$ and $L\in\mathbb{N}$, there exist 
$a>0$, $\lambda\in(0,1)$ and $b>0$ such that 
\begin{align*}
	K_{\varepsilon,L}(V)(q)\leq \lambda V(q)+b
\end{align*}
for all $q\in\mathbb{R}^d$. Note that we retain the same explicit Lyapunov function (\ref{eqn:explicit_lyapunov}) 
which will be needed in the following. The finite moment condition $\pi_0(V)$ holds for initial distributions with sufficiently 
light tails such as Gaussian distributions. 

\subsection{Excursions from convexity set in Theorem \ref{thm:relaxed_meeting}}
Since Assumption \ref{ass:convexity} is satisfied on any compact set, we can take $S=B_0(r)=\{q\in\mathbb{R}^d:|q|\leq r\}$ 
with $r>0$ arbitrarily large. Under the Lyapunov function (\ref{eqn:explicit_lyapunov}), the level set conditions in 
Theorem \ref{thm:relaxed_meeting} can be rewritten as 
\begin{align}\label{eqn:levelset_thm1}
B_0(a^{-1}\log\ell_1)\subseteq B_0(r)\cap L_{\ell_0}(U)
\end{align}
for some $\ell_0\in(\inf_{q\in S}U(q),\sup_{q\in S}U(q))$ and $\ell_1>1$ satisfying $\lambda+2b(1-\lambda)^{-1}(1+\ell_1)^{-1}<1$. 
With $\lambda\in(0,1)$ and $b>0$ fixed, the last inequality requires $\ell_1$ to be sufficiently large. 
As $\lim_{|q|\rightarrow\infty}U(q)=\infty$, the latter can be done without violating (\ref{eqn:levelset_thm1}) since we 
can choose an arbitrarily large $\ell_0$ by taking $r$ sufficiently large. 
This completes the verification of the assumptions required in Theorem \ref{thm:relaxed_meeting}.

\subsection{Remaining assumptions in Theorem \ref{thm:exact_meeting}}
For the Gaussian random walk kernel and the Lyapunov function (\ref{eqn:explicit_lyapunov}), it follows 
by a change of variables and the triangle inequality that 
\begin{align*}
	Q_{\sigma}(V)(q)=\int_{\mathbb{R}^d}\exp(a|q'|)\mathcal{N}(q';q,\sigma^2I_d)dq' 
				  &=\int_{\mathbb{R}^d}\exp(a|q+q'|)\mathcal{N}(q';0,\sigma^2I_d)dq' \\
				  &\leq \exp(a|q|)\int_{\mathbb{R}^d}\exp(a|q'|)\mathcal{N}(q';0,\sigma^2I_d)dq' 
\end{align*}
for any $\sigma>0$ and $q\in\mathbb{R}^d$. Therefore we can take $\mu=\int_{\mathbb{R}^d}\exp(a|q'|)\mathcal{N}(q';0,\sigma^2I_d)dq'$. 
With $\mu$ and $\lambda\in(0,1)$ fixed, we can define $\lambda_0=(1-\gamma)\lambda+\gamma(1+\mu)<1$ by taking $\gamma\in(0,1)$ 
small enough. Fixing also $\gamma$ and $b>0$, the level set conditions in Theorem \ref{thm:exact_meeting} hold using 
the same arguments as the previous section. This completes verifying the assumptions required in Theorem \ref{thm:exact_meeting}.
Lastly, we note that minor modifications of the above arguments would also show that the assumptions of Theorems \ref{thm:relaxed_meeting}--\ref{thm:exact_meeting} hold for any multivariate Gaussian target distribution.

\end{document}